\RequirePackage{fix-cm}
\documentclass{article}
\usepackage{authblk}

\usepackage{graphicx}
\usepackage[none]{hyphenat}

\usepackage{algorithm}
\usepackage{algpseudocode}      

\usepackage{latexsym}
\usepackage{amssymb}
\usepackage[T1]{fontenc}
\usepackage{geometry}
\usepackage{natbib}
\usepackage{amsmath}
\usepackage{graphicx}
\usepackage{hyperref}
\usepackage{enumitem}

\usepackage{xcolor}

\usepackage{amssymb}
\newcommand{\contaminated }{{\mathbb{I}_{C}}}

\usepackage{amsthm}

\newtheorem{proposition}{Proposition}%
\newtheorem*{proposition*}{Proposition}%

\usepackage{setspace}
\usepackage{array}
\newcolumntype{C}[1]{>{\centering\arraybackslash}m{#1}}
\begin{document}

\title{Modeling portfolio loss distribution under infectious defaults and immunization}

\author[a,b]{Gabriele Torri\thanks{Corresponding author \textit{
			gabriele.torri@unibg.it}}}
\author[a]{Rosella Giacometti}
\author[c]{Gianluca Farina}

\affil[a]{Department of  Management, University of Bergamo. Via dei Caniana, 2, 24127 Bergamo, Italy}
\affil[b]{Faculty of Economics, VSB-Technical University of Ostrava, 17. Listopadu 2172/15, 70800 Ostrava, Czech Republic}
\affil[c]{Mediobanca S.P.A. Milan, Italy}

\date{\today}
\maketitle
\begin{abstract}
We introduce a model for the loss distribution of a credit portfolio considering a contagion mechanism for the default of names which is the result of two independent components: an \textit{infection attempt} generated by defaulting entities and a \textit{failed defence} from healthy ones. We then propose an efficient recursive algorithm for the loss distribution. Then we extend the framework with more flexible distributions that integrate a contagion component and a systematic factor to better fit real-world data. Finally, we propose an empirical application in which we price synthetic CDO tranches of the iTraxx index, finding a good fit for multiple tranches. 
	
\end{abstract}

\section*{Acknowledgments} 
Gabriele Torri is grateful for financial support from the Czech Scientific Foundation (GACR), Grant Number 23-07128S.

\section*{Disclosure statement}
The authors report there are no competing interests to declare.

\section*{Data availability statement}
The data used in this study are provided by LSEG. Restrictions apply to the availability of these data, which were used under license for this study. 

\vfill
\pagebreak

\section{Introduction}

Modeling the joint loss distribution of a credit portfolio is a relevant and challenging task, fundamental for the pricing of credit derivatives. A vast literature that discusses the pricing of multiple-name credit derivatives (such as CDOs, i.e. Collateralized Debt Obbligations) flourished in the first decade of the 00s, period characterized by a strong expansion of such markets. Not surprisingly, the interest from the financial community declined after the Global Financial Crisis of 2007-08 due to the contraction of the CDOs markets, whose opacity and misleading pricing were considered among the main causes. Despite the reduced hype, and the tarnished reputation of CDO markets, synthetic CDOs contracts written on credit indices such as the CDX and iTraxx are still a relevant and growing market. These instruments allow investors to identify opportunities and hedging strategies. Regulators may also be interested in the information content of synthetic CDO tranches, as they embed the market estimates of default correlation and on the probability of extreme credit events \citep[see e.g.][]{ECB2006information,scheicher2008has,wojtowicz2014cdos,gourieroux2021disastrous}, providing market-based information relevant for the estimation of systemic risk (see e.g. \citealp{montagna2020origin,gourieroux2021disastrous}).

The primary contribution of this paper is to develop a novel and tractable model for constructing the loss distribution of a credit portfolio, providing also an analytical characterization of pairwise distributions among names. A secondary contribution is to design and empirically validate a recursive algorithm for deriving this distribution suitable for real-world asset pricing applications. Our approach extends the contagion-based framework of \cite{davis2001infectious}, but differs in key aspects. Instead of modeling pairwise contagion shocks, we assume that the default of a single name has systemic relevance, affecting the entire portfolio. We introduce an immunization mechanism that mitigates contagion effects and prevents feedback loops, simplifying calibration and improving computational efficiency. The model also relaxes the homogeneous portfolio assumption, allowing heterogeneous default probabilities and infection rates.
To enhance flexibility, we propose two specifications integrating contagion with a common factor: (i) contagion conditional on a Gaussian factor, and (ii) a mixture distribution. These formulations produce tractable loss distributions suitable for practical applications, with parsimonious parametrization. Finally, we validate our framework using synthetic CDO tranches on the iTraxx index, obtaining a good empirical fit and meaningful interpretation of the parameters.

The paper is structured as follows: Section 2 presents a review of the literature, Section 3 introduces our contagion model and describes the algorithm for the portfolio loss distribution. Section 4 extends the model to more realistic loss distributions that combine contagion dynamics and deafault correlations due to a common factor. In Section 5 a version of the model with a restricted parameter structure is applied to the problem of pricing CDOs, implemented then  on real-world iTraxx data in Section 6. Section 7 concludes, while proofs of theoretical results are provided in Appendix A.

\section{Literature review}

One of the main problems in credit modeling is the default clustering: it has been observed, especially during recessions, that defaults are not uniformly spaced in time but rather tend to concentrate over small periods of time (see e.g. \citealp{azizpour2018exploring} presenting evidence of default clustering in the US market). 

The dependence between obligors in a credit portfolio mainly arises from two sources of risk. The first source is a cyclical dependence to underlying common factors inducing systematic risk. The second source of risk is linked to the degree of interaction between obligors. Financial distress affecting a small group of obligors can spread to a large part of the economy inducing a rapid increase of defaults. We refer to this risk as default contagion or systemic risk. This component of risk is the channel of defaults acceleration through domino effects.

The most used approaches to introduce dependence among defaults (via exposure to common factors, by correlating the default intensity processes, and by direct application of copula methods) struggle to replicate the observed clustering pattern. 
\\
\cite{duffie2009frailty} empirically test  whether under doubly stochastic assumption the process of cumulative defaults can be modeled  as time changed Poisson process. They used  intensity estimated on US data of corporate defaults from 1979 to 2004 derived in \cite{das2007common}. Their result point out that the proposed model is not able to explain properly the default correlation of the data.  Similarly, \cite{azizpour2018exploring} find that the amount of default clustering cannot be explained by the exposure to observable and latent factors, showing strong evidence of contagion dynamics.

\cite{frey2001copulas} show that the correlation structure between obligors is not able to capture the dependence completely. It is necessary to introduce more relevant information on  the lower tail dependence. This is equivalent to exogenously specify the appropriate copula from which the amount of default contagion is derived.
\cite{yu2007correlated} works with an intensity based model where the default intensities are driven also by past defaults history, in addition to exogenous factors. A special case of this model is the copula approach presented by \cite{schonbucher2001copula}. Still, the direct modeling of the correlation between default intensity processes introduces a weak dependency resulting in a lower joint probability of defaults  than expected for highly correlated entities \citep{jouanin2001modelling}.  Other approaches focus on distribution with fat tails for the modeling of latent factors \citep[see e.g.][]{kalemanova2007normal,albrecher2007generic}, or model the evolution of latent variables of a structural model using jump-diffusion processes \citep{bujok2012numerical}.\\
One valid alternative  is to  explicitly enrich the modeling framework with contagion or infection mechanisms; this increases the probability of observing extreme losses in the portfolio and can also account for default clustering. Adding this particular dependency structure introduces a looping mechanism that makes calibration problematic: the probability of default of each name can impact and is impacted by the probability of default of the others. Several attempts have been proposed in order to resolve the looping issue when adding such effects.
Relative to contagion mechanisms, \cite{jarrow2001counterparty} impose a hierarchical approach and suggest to separate the firms into two groups: primary names can only default idiosyncratically while names in the second group can also default because of infection starting from the first set of names. Their work generalizes reduced-form models by making default intensities depend on counterpart default. The primary/secondary separation has been applied by other authors too, for example \cite{rosch2008estimating}: they start from a one-factor model where the number of defaults of primary names can affect the default probabilities of secondary names. \newline
In \cite{neu2004credit}, firms can have either mutually supportive or competitive relationships between them; a default will hence decrease the default probability of competitors but increase the same quantity for firms that had positive inflows from the name in distress (for example suppliers/clients).\newline
\cite{egloff2007simple} instead add micro-structural dependencies via a directed weighted graph and show how even well diversified portfolios carry significant credit risk when such inter-dependencies are accounted for. The open problem in their approach (as well as in other network based models) is the calibration of the weights that form the network.\newline
\cite{giesecke2004cyclical} add contagion processes to more standard common factor approaches while \cite{frey2010dynamic}  apply a Markov-chain model with default contagion to the problem of dynamically hedge CDO products. Nowadays the attention on contagion mechanise is increasing with a particular attention to explicit modeling on the interconnections among obligors (see \citealp{torri2018robust}).\\

Our work draws inspiration  from the seminal paper of \cite{davis2001infectious} where firms can default in two ways, either idiosyncratically or via infection from other defaulting names. Unfortunately, in their most generic (and therefore elegant) specification of the model, one has to rely on Monte Carlo simulations in order to obtain the portfolio loss distribution. This is not computationally efficient for large portfolios since in a portfolio of size $n$, the variables to consider are of order $n^2$. The solution proposed by the authors leads to closed form results for the loss distribution at the cost of quite strict assumptions: the portfolio considered has to be homogeneous with respect to the probability of idiosyncratic default, the infection rates and the losses given default (LGD).
\cite{sakata2007infectious} extended Davis and Lo's original model by assuming that idiosyncratic defaults might in fact be avoided with the help of non-defaulted names. In their model there is hence not only a contagion that causes more defaults but also a positive effect due to the intervention of other names, called recovery spillage, that prevents entities from defaulting. The main results they obtain are similar to the ones presented in \cite{davis2001infectious} in terms of both complexity and assumptions required (homogeneous portfolio).
\newline
\cite{cousin2013extension} use a multiple period model where defaults happening at time $t$ can still cause contagion later on. In addition, they consider the case where more than one infection is needed to cause a default by contagion and, from a theoretical point of view, they relax several of  \cite{davis2001infectious} assumptions. Unfortunately, the most generic specification of the model is  computationally very demanding and  the numerical applications shown in \cite{cousin2013extension}  are based on assumptions that are in line with \cite{davis2001infectious}; in particular, the authors require that the Bernoulli variables used are exchangeable. \newline

Overall, considering the literature on contagion mechanisms, we can identify two groups of contagion models used for constructing the loss curve. On one side, there are elegant models with permissive assumptions, though they are cumbersome to use due to the lack of closed-form solutions. On the other side, computationally efficient algorithms exist, but they rely on restrictive assumptions.\\
Ideally, one would like to maintain the contagion mechanism as general as possible but computationally tractable: the proposed model aims to provide a compromise between these two relevant aspects, that moves toward this ideal solution. 

\section{The proposed  infection model} \label{sec:contagion}
As previously mentioned, our model draws inspiration  from the seminal paper of \cite{davis2001infectious}.
Given a portfolio of $n$ entities, \cite{davis2001infectious}  model the default event at time $t$ of entity $i$ via the Bernoulli variable $Z_i(t)$ that can either take value 0 (indicating survival) or 1 (default) and that is constructed according to the following equation
\begin{equation} \label{eq:DavisLoModDesc}
	Z_i(t) = X_i(t) + (1-X_i(t)) \cdot \left[ 1 - \prod_{j\neq i} \left(1-X_j(t) \cdot Y_{i,j}(t)\right) \right].
\end{equation}
The main assumption behind  equation (\ref{eq:DavisLoModDesc}) is the existence of $n$ variables $X_i(t),\space i=1,\cdots, n$ and $n\cdot (n-1)$ variables $Y_{i,j}(t), \;i,j=1,\cdots,n,\; j\neq i$. Variable $X_i(t)$ is responsible for idiosyncratic default of the $i$ name, while $Y_{i,j}(t)$ governs the possibility that name $i$ is infected by name $j$. 
The additional assumption taken is that the variables $X_i(t)$ and $Y_{i,j}(t)$ are i.i.d. according to a Bernoulli distribution. \newline

It is a one period model $[0,t]$ where, at time $t$, firm $i$ is in default, (i.e. $Z_i(t) =1$), either  via an idiosyncratic default (i.e. $X_i(t)=1$) or if at least one other bond $j$ defaults directly and infects the first one ($X_j(t)=1$ and $Y_{i,j}(t)=1$). Dependency among the variables $Z_i(t)$ is hence introduced via the infection mechanism triggered by the $Y(t)$s variables. 

This very elegant and extremely flexible formulation has one main drawback: there are no closed-form formulas (or even semi-analytical techniques) that can be applied to determine the total portfolio loss distribution. The selective default mechanism induce an asymmetry  in the $Y_{i,j}$ variables that allow for name $j$ to selectively ``infect'' some names but not others. 
\newline
In order to solve this problem, the model we propose drops the above highly asymmetric framework. In fact, we introduce two important innovations:  
\begin{enumerate}
	\item each name, upon idiosyncratic default, can either infect no other name, or spread an infection attempt to the entire system instead of individual firms. The economic interpretation of the infection channel is that the idiosyncratic default of the firm is read, by the rest of the market, as a shock capable of triggering more defaults. 
	\item Other names can survive infection attempts via an immunization mechanism.
\end{enumerate}
We introduce a potential systemic propagation from the default of $j$ to all nodes and we allow for each node $i$ to defend itself. So from one side we have simplified the original model by dropping the ``selective'' infection mechanism and, on the other side, we have enriched it by allowing names to develop immunization.

The following equations model our approach:
\begin{equation} \label{eq:MainModDesc}
	Z_i(t) = X_i(t) + \left[1-X_i(t)\right] \cdot \left[1-U_i(t)\right] \cdot \left\{ 1 - \prod_{j\neq i} \left[1-X_j(t) \cdot V_j(t)\right] \right\},
\end{equation}
where we have postulated the existence of $3n$ mutually independent Bernoulli variables $X_1,...,X_n,V_1,...,V_n,U_1,...,U_n$ compared of the $n+n(n-1)$ of \cite{davis2001infectious}.

We  have two mechanisms for default: via idiosyncratic defaults $(i.e. X_i(t)=1)$ and via contagion. In particular name $j$ infects name $i$ only if two independent conditions are satisfied:
\begin{enumerate}
	\item \textbf{Infection attempt from $j$}: name $j$ defaults idiosyncratically and attempts to spread the infection to \textbf{all} other names. This component is driven by the $V(t)$ variables and $V_j(t)=1$ means that $j$ is infective ($X_j(t)=1$ and $V_j(t)=1$).
	\item \textbf{Failed defence from $i$}: name $i$ fails to defend itself ($U_i(t)=0$) from \textbf{every} possible infection. This component is driven by the $U$ variables.
\end{enumerate}
The independence assumption among the building blocks of \eqref{eq:MainModDesc} serves two purposes: firstly, it is crucial for proving theoretical results in later sections of the paper, as it allows to split joint distributions in an easy way. Secondly, it represents a tractable mechanism of creating dependency: we use specific combinations of independent variables to generate dependent ones, in exactly the same way it is done, for example, in factor models where independent building blocks (the common and the idiosyncratic factors) are assembled together to create dependency.\newline

The mutually independent Bernoulli variables $X_i(t),V_i(t),U_i(t)$, $i=1,...,n$, have  probability  $ P\{X_i(t)=1\} = p_i, \,\, P\{U_i(t)=1\} = u_i, \,\, P\{V_i(t)=1\} = v_i $ on a time horizon $[0,t]$. For the rest of the paper, to simplify the notation, we omit the time index when not required. 
For each name $i$, the set of parameters $[p_i, u_i,v_i]$ defines its behaviour in the model. In particular, high values of $v$ represent names that, upon default, are extremely infective. This can be used for pivotal names that are regarded as critical for the well being of the entire system. Via $v_i$ it is possible to tune the shock that the idiosyncratic default of entity $i$ has on the rest of the system. Low values of $u$ represent names that are strongly dependent on the health of the rest of the system. The exact economic interpretation of the immunization depends on the nature of the system we are trying to model. The role of $p$, instead, is to control only the probability of idiosyncratic default. We stress that $p_i$ is not the final probability of default of the name, as the latter (i.e. $\tilde{p}_i := P\{Z_i=1\}$) is the result not only of $p_i$ but also of the rate of infections from the other names as well as its own ability to benefit from immunization via $u_i$.\newline
We present a useful proposition that is proved in appendix \ref{A} and that provides $P\{Z_i=1\}$, the probability of default of a single name.

\begin{proposition}\label{pr: I}
	Let $X_i,V_i,U_i$, $i=1,...,n$ be mutually independent Bernoulli variables with probabilities  $p_i= P\{X_i=1\}, \,\, u_i=P\{U_i=1\} , \,\, v_i=P\{V_i=1\}$ on a time horizon $[0,t]$. Let $Z_i$ be defined according to Eq. \ref{eq:MainModDesc} and let $\tilde{p}_i := P\{Z_i=1\}$. We have:
	\begin{equation}\label{eq: PD} \tilde{p}_i = p_i+[1-p_i]\cdot [1-u_i]\cdot I_{\overline{\{i\}}},\end{equation}
	where 
	\begin{equation} I_{\overline{\{i\}}} := \left\{1-\prod_{j\neq i } \left[1-p_j\cdot v_j\right]\right\}. \end{equation}
\end{proposition}

The above result shows that the marginal default probability of entity $i$ is a function of its idiosyncratic default probability, its immunization ability and the component $I_{\overline{\{i\}}}$. It is worth noting that it \textit{does not} depend on the marginal default probabilities of the other names and neither on their immunization capabilities: the looping mechanism that entangles marginal calibrations in other approaches has been effectively broken.

\subsection{Portfolio loss distribution at time $t$} \label{sec:loss_distr}
Let $L_n(t)$ represents the total amount of losses at time $t$ of a portfolio with $n$ names:
$$L_n(t)=\sum_{i=1}^{n} d_i \cdot Z_i(t),$$
where the integer numbers $d_i$ are the units of losses associated with the default of the $i$ name and $Z_i(t)$ is the random variable registering the default on a time horizon $[0,t]$.
The loss distribution $L_n(t)$ is affected by the probabilities $ P\{X_i=1\} = p_i(t), \,\,$ and $ P\{U_i=1\} = u_i(t), \,\, P\{V_i=1\} = v_i(t) $ on the time horizon $[0,t]$.

We present an algorithm for calculating $P\{L_n(t)=h\}$ for a given integer $h$. The algorithm is similar, in spirit, to the one presented by \cite{andersen2003all} for conditionally independent models. 
\cite{andersen2003all} showed that an efficient way of performing the convolution of the independent conditioned default probabilities is to construct the portfolio loss distribution by adding each name $j$ with $j=1,...,n$ one by one via a recursive relationship. 

In our model, we can exploit the independence components in equation (\ref{eq:MainModDesc}) to obtain a recursive algorithm.  The rest of this section is devoted to the introduction of our procedure to construct the portfolio loss distribution.\newline
The amount of the portfolio loss can be seen as the sum of two components:
$$
\begin{array}{c}
	L_n(t)=\sum_{i=1}^{n} d_i \cdot Z_i(t) =
	L_n^I(t) + L_n^C(t).
\end{array}
$$
where $L_n^I(t)$ and $L_n^C(t)$ represents the units of losses due to idiosyncratic effects and to contagion events, respectively. Beyond these two components, we consider a third group: firms that, in an uncontaminated world, have neither activated defensive strategies nor defaulted idiosyncratically. These firms remain solvent only because contagion has not yet occurred; in the presence of contagion, they would instead default.
We therefore introduce a third quantity $L_n^P(t)$ which denotes the number of units of potential losses in a (so far) uncontaminated environment, corresponding to firms that are not in default but have not adopted any defensive strategy.
Formally, these three building blocks are defined as:
$$
\begin{array}{lr}
	L_n^I(t):=\sum_{i=1}^{n} d_i \cdot X_i(t)  & \text{Idiosyncratic driven losses} \\
	L_n^C(t):=\sum_{i=1}^{n} d_i \cdot (Z_i(t)-X_i(t)) & \text{Contagion driven losses}\\
	L_n^P(t) := (1-\contaminated)\cdot \sum_{i=1}^{n} d_i \cdot [1-X_i(t)]\cdot [1-U_i(t)] & \text{Potential losses}
\end{array}
$$
where $\contaminated$ is an indicator function that is equal to one if there is at least one infection active and zero otherwise:  
\begin{equation}
	\contaminated = \begin{cases}1, & \text{if}\; \sum_{i = 1}^n X_i V_i >0 \\ 0, &\text{otherwise} \end{cases}.
\end{equation}

We underline that if contagion has happened in the system (i.e. $\contaminated=1$), $L_n^P(t)$ is equal to zero, since any name without a defensive strategy is affected by the infection and thus they are counted in the $L_n^C(t)$ component. In contrast, $L_n^P(t)=0$ does not imply that there are no active contagions, as it can be that a company $j$ transmitted contagion, but any other company $i$ either defaulted idiosyncratically ($X_i = 1$), or has immunity (i.e. $U_i = 1$).

The distribution of $L_n(t)$ can be described as:
\begin{equation} \label{eq:LossFromAlfaBeta_prep}
	\begin{array}{lr}
		P\{ L_n(t) = h\} &= \\
		P\{ L_n^I(t)+L_n^C(t) = h\} &= \\
		\sum_{{k=0}}^{\bar{\ell}} \left[ P\{ L_n^I(t)+L_n^C(t) = h | L_n^P(t)=k\} \cdot P\{L_n^P(t)=k\} \right]&=\\
		\sum_{{k=0}}^{\bar{\ell}} \left[ P\{ L_n^I(t)+L_n^C(t) = h, L_n^P(t)=k\} \right], &
	\end{array}
\end{equation}
where the integer $\bar{\ell}=\sum_n d_i$ is the maximum amount of losses in the system. We further partition the probability space on the basis of the indicator function $\contaminated$, and discarding the events with zero probability we rewrite Equation (\ref{eq:LossFromAlfaBeta_prep}) as follows:

\begin{equation} \label{eq:LossFromAlfaBeta}
	\begin{array}{ll}
		P\{ L_n(t) = h\} = & \sum_{{k=0}}^{\bar{\ell}-h} \left[ P\{ L_n^I(t)=h,L_n^C(t) = 0, L_n^P(t) = k,{\contaminated = 0}\}\right] \,+\\
		\quad & \sum_{{k=0}}^{{h}} \left[P\{ L_n^I(t)=k,L_n^C(t) = h-k, L_n^P(t) = 0, {\contaminated = 1}\} \right],
	\end{array}
\end{equation}

where the first term refers to the cases with no active infections (uncontaminated world), and the second term to the case in which there is at least one active infection (contaminated world). Indeed, to have $h$ losses, either (a) we are in an uncontaminated world with $h$ units of idiosyncratic losses (and $k$ units of potential losses), or (b) we are in a contaminated world with $k$ units of idiosyncratic driven losses and $h-k$ units of contagion driven losses. We underline that the potential losses in the contaminated world are equal to 0 by construction: in such state, all the possible losses due to contagion have manifested ($L_n^C(t) = h-k$) and no possible other contagion losses are possible ($L_n^P(t) = 0$). Thus, any state with $L_n^P(t) \neq 0$ and $\contaminated = 1$ has zero probability and is not considered.
Let us finally define for convenience two quantities:
\begin{eqnarray*}
	\alpha_{n}(h,k,t) & := & P\{L_n^I(t)=h,\: L_n^C(t)=0,\: L_n^P(t)=k,\contaminated = 0\}, \\
	\beta_{n}(h,k,t) & := & P\{L_n^I(t)=h,\: L_n^C(t)=k,\: L_n^P(t)=0,\contaminated = 1\},
\end{eqnarray*}

and Equation (\ref{eq:LossFromAlfaBeta}) can be written as
$$ P\{ L_n(t) = h\} =\sum_{k=0}^{\bar{\ell}-h} \alpha_n(h,k,t) + \sum_{k=0}^h \beta_n(k,h-k,t). $$

$\alpha_n(h,k,t)$ represents the probability of realizing $h$ units of losses in an uncontaminated world of $n$ names, in which there are also $k$ units of losses at risk should an infection appear. On the other hand, $\beta_n(h,k,t)$ represents the probability of realizing $h+k$ units of losses in a contaminated universe of $n$ names of which $h$ are due to idiosyncratic defaults and $k$ are due to pure infection. We point out that the lack of contagion driven losses and potential losses (i.e. $L_n^C(t)=0$, $\: L_n^P(t)=0$) does not imply that there are no active infections: indeed we have the cases in which a contagious default occurs, but it does not affect any other company, as other companies are either defaulted idiosyncratically ($X_j = 1$), or are immune to infection ($U_j = 1$). For this reason $\alpha_{n}(h,0,t) \neq \beta_{n}(h,0,t)$.\newline

In order to construct recursively the portfolio loss distribution we formulate the following proposition. A derivation of this system of equations can be found in appendix \ref{A}. When not necessary, we omit the argument $t$ in order to simplify the notation.\newline

\begin{proposition}\label{portfolio}
Consider a portfolio comprising $n$ assets, with each asset denoted by the index $j = 1,\dots,n$. The following recursive relationship links $[\alpha_j(\cdot,\cdot), \beta_j(\cdot,\cdot)]$ to $[\alpha_{j-1}(\cdot,\cdot),\beta_{j-1}(\cdot,\cdot)]$:
	\begin{equation}\label{alfabeta}
		\begin{array}{cclr} 
			\ 								&\  & (1-p_j) \cdot u_j \cdot \alpha_{j-1}(h,k) 			& + \\
			\alpha_j(h,k) & = & (1-p_j) \cdot (1-u_j) \cdot \alpha_{j-1}(h,k-d_j) & + \\ 
			\ 								&\  & p_j \cdot (1-v_j) \cdot \alpha_{j-1}(h-d_j,k),			& \  \\
			\ 								&\  & \ 											& \  \\
			\ 								&\  & (1-p_j) \cdot u_j \cdot \beta_{j-1}(h,k) + p_j \cdot \beta_{j-1}(h-d_j,k) & + \\
			\beta_j(h,k) 		& = & (1-p_j) \cdot (1-u_j) \cdot \beta_{j-1}(h,k-d_j) & + \\ 
			\ 								&\  & p_j \cdot v_j \cdot \alpha_{j-1}(h-d_j, k), &\ \\
		\end{array}
	\end{equation}
	with the following boundary conditions:
	\begin{equation}\label{boundarycond}
		\begin{array}{l}
			\alpha_{0}(0,0) = 1, \\ \alpha_{0}(i,j) = 0 \quad \; \forall (i,j) \neq (0,0), \\ \beta_{0}(i,j) =0 \; \quad \forall i,j.
		\end{array}
	\end{equation}
	Moreover, the order chosen to include the terms does not affect the final result.
\end{proposition}

Proposition \ref{portfolio} allows to construct the loss distribution recursively, adding progressively one name at the time (that is, we first consider a system with a single asset, then with two assets, and so on up to $n$). Since the recursion does not depend on the order chosen when adding names, the final distribution of the portfolio losses defined by $\alpha_n(h,k)$ and $\beta_n(k,h-k)$ is not affected by the order of the companies. We also point out that if the portfolio is equally weighted (as in the case of synthetic CDOs),  and the loss given default (LGD) are assumed to be uniform across names, we can set $d_j = 1 \; \forall j$, so that $\bar{\ell} = n$, and then rescale accordingly the support of the loss distribution.

In practice, assuming constant LGD, we set $d_j = 1 \; \forall j$ and we compute the coefficients $\alpha_n(h,k)$ and $\beta_n(h,k)$ for Equation \ref{eq:LossFromAlfaBeta} using Algorithms \ref{algo:alpha} and \ref{algo:beta}. Both algorithms have cubic space and time complexity, although many combinations are skipped. The implementation is practical for applications with a moderate number of assets (up to $200$, as in the case of synthetic CDOs such as the iTraxx, where $n=125$), but for larger portfolio the algorithm becomes impractical. In such cases, a Monte Carlo approach may be suitable, as discussed in Appendix \ref{sec:MC}.
 
\begin{algorithm}[H]
	\caption{Computation of $\alpha_n(h,k)$ coefficients} \label{algo:alpha}
	\begin{algorithmic}[1]
		\Require Vectors $p$, $u$, $v$ of length $n$
		\For{$j = 0$ to $n$}
		\For{$h = 0$ to $n$}
		\For{$k = 0$ to $n$}
		\If{$j=h=k=0$}
		\State $\alpha_0(0,0) \gets 1$
		\ElsIf{$j=0$ \textbf{or} $h>j$ \textbf{or} $k>j-h$}
		\State $\alpha_j(h,k) \gets 0$
		\Else
		\State $\alpha_j(h,k) \gets (1-p_j)u_j\,\alpha_{j-1}(h,k)$
		\If{$k>0$}
		\State $\alpha_j(h,k) \gets \alpha_j(h,k) + (1-p_j)(1-u_j)\alpha_{j-1}(h,k-1)$
		\EndIf
		\If{$h>0$}
		\State $\alpha_j(h,k) \gets \alpha_j(h,k) + p_j(1-v_j)\alpha_{j-1}(h-1,k)$
		\EndIf
		\EndIf
		\EndFor
		\EndFor
		\EndFor
	\end{algorithmic}
\end{algorithm}

\begin{algorithm}[H] 
	\caption{Computation of $\beta_n(h,k)$ coefficients} \label{algo:beta}
	\begin{algorithmic}[1]
		\Require Vectors $p$, $u$, $v$ of length $n$, and array $\alpha_j(h,k)$
		\For{$j = 0$ to $n$}
		\For{$h = 0$ to $n$}
		\For{$k = 0$ to $n$}
		\If{$j=0$}
		\State $\beta_0(h,k) \gets 0$
		\ElsIf{$h + k > j$}
		\State $\beta_j(h,k) \gets 0$
		\Else
		\State $\beta_j(h,k) \gets (1-p_j)u_j\,\beta_{j-1}(h,k)$
		\If{$k>0$}
		\State $\beta_j(h,k) \gets \beta_j(h,k) + (1-p_j)(1-u_j)\beta_{j-1}(h,k-1)$
		\EndIf
		\If{$h>0$}
		\State $\beta_j(h,k) \gets \beta_j(h,k) + p_j\,\beta_{j-1}(h-1,k) + p_jv_j\,\alpha_{j-1}(h-1,k)$
		\EndIf
		\EndIf
		\EndFor
		\EndFor
		\EndFor
	\end{algorithmic}
\end{algorithm}

The following result ensures instead that the probability of observing no losses is a (decreasing) function only of the $p_i$. The formal proof can be found in appendix \ref{A}. However the intuition behind the result is straightforward as the only way we experience no losses is that every name survives idiosyncratically (default by contagion requires at least one idiosyncratic default).
\begin{proposition}\label{pr:NoLosses}
	The probability of observing no losses is given by the following result:
	\begin{equation} \label{eq:NoLosses}
		P\{ L_n (t)= 0\} = \prod_{j=1}^{n} (1-p_j(t)).
	\end{equation}
\end{proposition}
Using the above result, we can be sure that the probability of observing no losses will decrease in time if the probabilities of idiosyncratic default in the two specifications of the model are increasing i.e. $$t_1<t_2, \quad p_i(t_1)\leq p_i(t_2),\quad \forall i \Longrightarrow P\{ L_n(t_1) = 0\} \geq P\{ L_n(t_2) = 0\}.$$

\section{Portfolio loss distribution with contagion and correlations} \label{sec:mixture_model}
In the proposed model, contagion is  the only source of dependence between defaults. More realistically, defaults can present a dependence structure due not only to contagion, but also to the effect of common factors affecting the value of the companies (i.e. we introduce a systematic component), and consequently the default probabilities. We develop here two approaches for the construction of a realistic loss distribution for a credit portfolio. The first is a natural extension of the well known one factor Gaussian model, where contagion is applied conditionally to the effect of a common systematic factor. The second model considers a mixture distribution, under the assumption that two states of the world may manifest in the future, one in which defaults are correlated due to a common Gaussian factor, and another in which dependency is due to contagion. Both methods will be considered in the empirical analysis for the pricing of synthetic CDOs.

\subsection{Conditional contagion loss distribution} \label{sec:cond_contagion}
In the first approach the defaults are driven by a common latent factor, and contagion is applied on top of the effect of the factor. 
We consider a credit portfolio in which default risk is modeled through a structural Vasicek-type framework \citep{vasicek2002distribution}. For each obligor $ i$, the default occurs between time $0$ and $t$ when a latent variable $Z_i $ falls below a fixed threshold $ \theta_i $. The latent variable follows the standard one-factor Gaussian model:
\[
Z_i = \sqrt{\rho} \, Y + \sqrt{1 - \rho} \, \varepsilon_i,
\]
where $ Y \sim \mathcal{N}(0,1) $ is a systematic risk factor common to all obligors, $ \varepsilon_i \sim \mathcal{N}(0,1) $ are idiosyncratic shocks independent across $ i $ independent of $Y$, and $\rho$ is the correlation between each couple of latent variables $Z_i$, $Z_j$. Before considering contagion --  defaults are independent conditional on the realization of $ Y = y $, and the default probability for obligor $ i $  between times $0$ and $t$ becomes
\[
p_i(t)|y = \mathbb{P}(Z_i \leq \theta_i \mid Y = y) = \Phi\left( \frac{\theta_i - \sqrt{\rho} \, y}{\sqrt{1 - \rho}} \right),
\]
where $ \Phi $ denotes the standard normal cumulative distribution function.

To capture systemic contagion effects, we extend this model by allowing defaulted entities to transmit distress to others. While defaults remain conditionally independent given  $Y $, once defaults occur, they can trigger additional defaults through the contagion mechanism presented in Section \ref{sec:contagion}, dependent on the initial set of defaults. 
In practice, we discretize the Gaussian distribution of $Y$. Specifically, we use a Gauss-Hermite quadrature and we approximate the unconditional expectation over $Y$ as
\[
\mathbb{E}[f(Y)] \approx \sum_{j=1}^m w_j f(y_j),
\]
where $y_j$ and $w_j$ are the nodes and weights of the $m$-point Gauss-Hermite quadrature rule. The contagion-adjusted conditional loss distribution $P\{L_n^{(COND)}(t)=h\mid Y = y_j\}$ is then computed for each discretized state $ y_j $ of the systematic factor: for each $y_j$, we compute the conditional default probabilities $p_i(t)\vert y_j$, run the contagion model to obtain the conditional loss distribution, and finally integrate over the values of $Y$ to obtain the unconditional loss distribution:
\[
P\{L_n^{(COND)}(t)=h\} \approx \sum_{j=1}^m w_j P\{L_n^{(COND)}(t)=h\mid Y = y_j\}.
\]

Since the contagion model has to be evaluated for every value of the conditioning variable $Y$, the computational burder scales with the number of nodes in the Gauss-Hermite quadrature. In the empirical application, balancing the quality of the approximation and the computational burden, we set $m =10$.

\subsection{Mixture contagion loss distribution} \label{sec:mix_contagion}
The second approach for the construction of the loss curve uses a mixture distribution for the portfolio losses to account for both the contagion mechanism and default correlations induced by a common Gaussian factor affect the loss distribution. The idea is to consider a system with two possible regimes: between times $0$ and $t$, we assume that the world is in any of two states: a ``contagion state'' with probability $\pi$, and a ``common factor state or correlated default state'' with probability $1-\pi$. In the former, defaults happen either  idiosyncratically, or according to the contagion described in Section \ref{sec:loss_distr}. In the latter, defaults are driven by a traditional one factor Gaussian model. From an economic standpoint, the approach is consistent with a market in which the two potential states can manifest, players are uncertain about the possible future state of the world, and assign a probability to each of them.

This approach allows us to model in a simple and flexible way the loss distribution, and it offers a straightforward structural interpretation for both the ``contagion'' and ``correlated default state''. On the flip side, the simplification of assuming independence between the two states does not allow to account for the interaction effects, failing to model for instance an increase of the contagion risk as a consequence of the negative movement of a common market factor. 

Formally, we model the random variable for the losses $L_n^{(MIX)}(t)$ as:

\begin{equation}
	L_n^{(MIX)}(t) = \xi L_n^{(CON)}(t) + (1-\xi) L_n^{(OFG)}(t), \label{eq:mix_model}
\end{equation}

where $L_n^{(CON)}(t)$ and $L_n^{(OFG)}(t)$ are the losses variable in the contagion state, and in the correlated default state, respectively, and $\xi$ is a Bernoulli variable such that $\xi=1$ with probability $\pi$ and $\xi=0$ otherwise. Due to independence between $L_n^{(MCM)}(t)$ and $L_n^{(OFG)}(t)$, the distribution function is a mixture distribution:
\begin{equation}
	P\{L_n^{(MIX)}(t)=h\} = \pi P\{{L}_n^{(CON)}(t)=h\} + (1-\pi) P\{{L}_n^{(OFG)}(t)=h\}.
\end{equation}

In this work, for the ``correlated default state'' we opted for a simple one factor Gaussian model as it is  easy to calibrate, and well known among regulators and practitioners, but we could choose an alternative model that accounts for fatter tails such as the one proposed in \cite{kalemanova2007normal} that assumes a NIG distribution if the fit of the model is not satisfactory. As discussed in Section \ref{sec:calibration}, we consider a restricted version of the model with three parameters that allows us to satisfactory calibrate the quotes of iTraxx Synthetic CDO index tranches. 

We underline that mixture models are commonly used in the modelization of portfolio credit risk (see e.g. \textit{CreditRisk+}, \citealp{suisse1997creditrisk+}), although in a different way: \cite{frey2003dependent} extensively discuss Bernoulli mixture models in which defaults present a conditional independence structure conditional on common factors, showing how latent variable models such as KMV \citep{kealhofer2001portfolio} and RiskMetrics \citep{morgan1997creditmetrics} can be mapped to equivalent mixture models. In such cases, the conditional distribution of portfolio losses is integrated on the value of the conditioning factors. In our approach, instead, the loss distribution in each of the two possible states of the world (``contagion'' and  ``correlated default'') is integrated on the distribution of the mixing binary variable $\xi$. 

Conceptually, our two-state approach is related to the literature on regime-switching models, that are statistical frameworks in which parameters or dynamics change according to an unobserved state variable that evolves over time, typically governed by a Markov process. such models have been used in systemic risk analysis (see e.g. \citealp{abdymomunov2013regime,bernardi2019switching}) to address some observed market characteristics. In particular, the literature documents the presence of structural breaks in the risk transmission in crises periods (\citealp{billio2012econometric}) and default clustering \citep{berndt2010correlation,azizpour2018exploring,montagna2020origin}. Regime switching models have also been widely used in credit risk measurement \citep[see e.g.][]{davies2004credit,liang2012reduced}. An approach similar to ours is the one proposed by \cite{davidson2020interdependence}, who develops a Bayesian time-varying parameter VAR model to distinguish between interdependence, contagion via interdependence, and abrupt contagion using a model-switching framework in the context of Latin American crises. Although her focus is on macro‑financial linkages, this typology shares the conceptual intuition behind our two‑regime structure. In contrast, our model explicitly differentiates default dynamics through a single-factor model during peace periods and a contagion‑driven model during crises.

\subsection{Default correlations} \label{sec:def_correl}
To further describe the models, we propose here a homogeneous example and we characterize analytically the pairwise default correlation between any two assets.

\begin{proposition} \label{prop:correl}
	 Consider a portfolio with $n$ assets in which defaults are regulated by the infection model in Section \ref{sec:contagion} and such that $p_i = p,\; v_i = v,\; u_i = u , \tilde{p}_i = \tilde{p}\; \forall i$. The pairwise default correlation of $Z_i$ and $Z_j$ is:
	
\begin{equation}
	\rho_{CON}(Z_1,Z_2) = \frac{P(Z_i = 1,Z_j=1) -\tilde{p}^2}{\tilde{p}(1-\tilde{p})},
\end{equation}
where $\tilde{p}$ is the marginal default probability of each name (see Eq. \ref{eq: PD}), and the probability of joint defaults of asset with indices $i$ and $j$ is:
\begin{equation}\label{eq:joint_def}
	P(Z_i = 1,Z_j=1) = p^2 + 2p(1-p)(1-u)\left[1-(1-v)(1-pv)^{n-2}\right] + (1-p)^2(1-u)^2\left[1-(1-pv)^{n-2}\right].
\end{equation}	
\end{proposition}

The result is proved in Appendix \ref{A}.

Using the law of total covariance we can then compute the correlation between defaults also for the conditional contagion (COND) and mixture contagion (MIX) models (Section \ref{sec:mix_contagion}). For conditional contagion model the covariance between two defaults is:

\begin{align}
	\mathbb{C}_{COND}(Z_i, Z_j) & = \mathbb{E}\left[\mathbb{C}(Z_i, Z_j | Y) \right] + \mathbb{C}\left(\mathbb{E}[Z_i | Y], \mathbb{E}[Z_j | Y]\right)
\end{align}

where the first term is computed averaging the covariances of the contagion model in each of the conditioning state of the conditioning variable $Y$, and the second term corresponds to the variance of the expected value of defaults, since the default probabilities are uniform across the portfolio.

Pairwise default correlation can then be computed by dividing by the standard deviations of the univariate distributions:
\begin{equation}
	\rho_{COND}(Z_1,Z_2) = \frac{\mathbb{C}_{COND}(Z_i, Z_j)}{\sqrt{\sigma_{\tilde{p}_i} \sigma_{\tilde{p}_j}}} = \frac{\mathbb{C}_{COND}(Z_i, Z_j)}{\tilde{p}(1-\tilde{p})}.
\end{equation}

The default correlation of the mixture model in Section \ref{sec:cond_contagion} can be computed in a similar way, conditioning on the mixing variable $\pi$.

\section{An application to CDO pricing}
In this section we test our approach by pricing synthetic CDO contracts. Handling credit portfolio products requires the ability to model and calculate the entire portfolio loss distribution at several time steps. This application will hence show the tractability of the system of equations \eqref{alfabeta}-\eqref{boundarycond}, and the quality of fit of the contagion, the conditional, and the mixture models, in comparison with the one factor Gaussian model that serves as a baseline benchmark.

\subsection{CDO pricing}
Let us start by giving few basic information regarding CDO products, in particular synthetic ones. A synthetic CDO contract is a complex structured credit product that is written between two parties, the protection buyer and the protection seller. Differently from cash CDOs, that use mortgages and bonds, synthetic CDOs use CDSs instruments. The two parties will exchange cash payments based on the survival/default events occurring in a pool of underlying entities. The pool is sliced into tranches characterized by an attachment and detachment point that determine the subordination of the tranche. The notional of a given tranche starts to get eroded when the total losses suffered in the pool -- due to the underlying names defaulting -- are above the tranche attachment. The protection seller receives periodical payments calculated as a premium (spread) multiplied by the amount of outstanding notional left in the tranche at the payment time. In return for the premium payments, the protection seller has to compensate the protection buyer for the losses occurred in the tranche. A CDO trade has hence two legs; the coupon leg represents the expected value of the payments made by the protection buyer while the offsetting leg is the expected value of the default payments made by the protection seller and it is referred to as the protection leg. \newline

Let us define, for a CDO tranche with attachment and detachment $(a, b)$, the amount of outstanding notional left at time $t$ as the following quantity $S(a, b, t)$:
\begin{equation}S(a, b, t) := (b-a) \cdot P\{L_n(t) \leq a\} + \int_{a}^{b} (b-x) \cdot P\{L_n(t) = x\} dx, \end{equation}
where $L_n(t)$ represents the portfolio losses at time $t$. The CDO pricing formulas can be expressed in terms of $S(a, b, t)$:
\begin{equation}\label{eqCpnLegForm} CpnLeg = \text{cpn} \cdot \sum_i dcf(t_i, t_{i+1}) \cdot D(t_{i+1}) \cdot S(a, b, t_{i+1}),\end{equation}
\begin{equation}V_{CDO} = CpnLeg - DfltLeg,\end{equation}
\begin{equation}DfltLeg = \int_0^M -\frac{\partial S(a, b, t)}{\partial t} \cdot D(t) dt, \end{equation}

where $M$ is the trade maturity, $D(t)$ is the discount factor at time $t$, $dcf(t, s)$ is the day count fraction between the coupon dates $t$ and $s$ and the sum over $i$ in \eqref{eqCpnLegForm} is intended over the scheduled coupon payment dates.  Note how the term $\frac{\partial S(att, det, t)}{\partial t}$ represents the losses occurred at time $t$. Note also that we wrote the value of the CDO from the protection buyer point of view. Finally, the par spread of the trade is calculated as the fair value of the $cpn$, i.e. the level of $cpn$ that makes zero the present value of the trade. The ability to calculate $L_n(t)$ (and hence $S(a, b, t)$ for every tranche seniority structure) in an efficient way is crucial in order to calculate the value of a CDO trade.\newline

\subsection{Specifications and calibration of the loss distribution}\label{sec:calibration}

In the empirical analysis we compare four classes of models for the estimation of the loss distribution: 
\begin{itemize}
	\item \textbf{One factor Gaussian model (OFG)} in which default correlations are modelled using a single Gaussian factor;
	\item \textbf{Contagion models (CON)} in which defaults are either idiosyncratic or due to contagion;
	\item \textbf{Conditional models (COND)} in which contagion is applied on top of a single factor model (contagion conditional to correlated defaults).
	\item \textbf{Mix models (MIX)} in which two states of the world can manifest (contagion and correlated defaults).
\end{itemize}

The OFG model is a single-parameter model widely used in the industry for its simplicity and tractability. It is also considered a \textit{lingua franca} among practitioners: the OFG base correlation is used on CDO trading desks in much the same way the Black-Scholes implied volatility is used by option traders. Readers unfamiliar with the OFG model might refer to \cite{li2000default} for the original description or to \cite{burtschell2009comparative} for a good review of the methodology. Let just remember here that in the OFG framework a single correlation input $\rho\in[0,1]$ describes the likelihood of names to default together. In practice, though, it is well known that the model with a single parameter $\rho$ cannot explain the prices of the many tranches one observes on the market.\footnote{The standard solution is to resort to the base correlation approach. Every tranche is priced as a difference of two equity tranches, i.e. tranches with attachment set at zero (for example, a $3\%-6\%$ tranche will be valued as the difference of a $0\%-6\%$ tranche minus a $0\%-3\%$ one). These two valuations are treated as independent from each other and hence different specification of the model are used (i.e. different values of $\rho$ are employed).  
	This approach, known as base correlation, has the drawback of creating the potential for mispricing for non-standard tranches.} For simplicity we assume uniform default probabilities, that are set equal to the average of the default probabilities $\tilde{p}_i(t)$ estimated from the single name CDS spreads of the index contituents.

Concerning the Contagion models (CON), the full implementation of the model in Equation (\ref{eq:MainModDesc}) would require the calibration of $3n$ parameters for any time $t$ at which cash flows are exchanged, and would be impractical. We propose thus a restricted version of the model with only a single calibrated parameter $\omega$ that controls how much probability of default comes from idiosyncratic effects versus contagion, and $n$ fixed parameters $\mu_i, \, i = 1,\dots, n$ to account for the different contagion potential of firms in the index. In particular, we assume that the quantities $p(t)$, $v(t)$ and $u(t)$ are linked via the following relationships to the marginal default probabilities $\tilde{p}_i(t)$, the estimated probability of default of name $i$:
\begin{equation}
	\begin{array}{l} \label{eq:calibration}
		p_i(t) = (1-\omega) \cdot \tilde{p}_i(t), \\
		v_i(t) = \mu_i \cdot [1-\sqrt{\tilde{p}_i(t)}], \\
		u_i(t) = 1- \frac{ \tilde{p}_i(t) - p_i(t)}{[1-p_i(t)] \cdot I_{\overline{ \{i\} } } }.
	\end{array}
\end{equation}
The coefficients $\omega$ and $\mu_i$ are assumed to be constant for all the times $t$ at which cash flows are exchanged, and the proposed relationship are motivated by economic considerations:
\begin{itemize}
	\item The individual probability of default contains information on both idiosyncratic and contagion risk, in a proportion regulated by the parameter $\omega$. In particular, $\omega=0$ represents the no-contagion case, while a higher value for $\omega$ causes most of the losses to derive from contagion events. Note that we need $\omega < 1$ as we always need at least one initial idiosyncratic event to trigger contagion effects;
	\item The quantities $v_i(t)$ are in a monotonic decreasing relationship with $\tilde{p}_i(t)$; this reflects the fact that healthier firms have a bigger impact in case of idiosyncratic default than riskier ones; the market is expecting default of risky firms (hence the high probability of default) and therefore the shock when the event finally happen is minor. We  restrict the parameter $\mu_i$ to be either constant across all the names, or differentiated according to industrial sector, obtaining three variants of the model (see below); the square-root transformation applied to $\tilde{p}_i(t)$ amplifies variation in the lower range of values and redistributes mass across the support. This choice performs well empirically and alternative concave transformations  yield similar results;
	\item The quantities $u_i(t)$ are chosen according to \eqref{eq: PD} in order to satisfy marginal constraints. It is worth noting that we get the calibration to marginal info embedded in our multivariate model almost by construction;
	\item The idiosyncratic probability of default $p_i(t)$ is, for every $i$, directly proportional to $\tilde{p}_i(t)$ and hence it is increasing in time. Thanks to the comments made after proposition  \ref{pr:NoLosses}, with this choice we also guarantee that the probability of observing no losses is a decreasing function of time;
\end{itemize}
We underline that this specification of the model is still handling heterogeneous portfolios (the marginal probability of defaults of the names are different from each other). With above choices we effectively managed to reduce the number of parameters to one: $\omega$.

For the conditional model (COND) we discretize the distribution of the conditioning factor in 10 buckets, and we calibrate two parameters, assumed constant for every $t$: $\rho$ (i.e. the correlation in the latent variable for the one factor Gaussian model) and $\omega$ (representing the proportion of default probability due to contagion). For each value of the discretized distribution of the factor we compute $p_i(t)$, $v_i(t)$, and $u_i(t)$ as in Eq. \ref{eq:calibration} in such a way that the proportion between the conditional default probability due to contagion is always equal to $\omega$.

Finally, the mix model (MIX) has three calibrated parameters: $\rho$ and $\omega$ for the  ``correlated defaults'' and ``contagion states'', respectively, and the mixing probability $\pi$ (see Eq. \ref{eq:mix_model}).

\subsubsection{Choice of the parameters $\mu_i$}\label{sec:choice_mu}

As discussed above, we limit the calibration of the contagion model to one parameter, while maintaining a certain degree of heterogeneity and using the information relative to the idiosyncratic default probabilities $\tilde{p}_i(t)$. We can further tune the model by restricting the values of $\mu_i$ depending on economic considerations. We consider three alternatives: first we set the value constant for all the names in the index (model CON-FLAT). In particular, we set the value to $\mu_i = 0.1$ for all $i$. We then consider two alternative calibrations that attribute a larger potential for the transmission of contagion to banking and financial institutions, as the literature on systemic risk highlights the centrality of such sectors in the amplification of risk and the transmission of financial distress to the whole economic system \citep{schwarcz2008systemic,freixas2015systemic,montagna2020origin}. In particular, we consider the model CON-BNK that set the coefficient $\mu_i=0.2$ for companies in the banking sector and $\mu_i = 0.05$ for other companies, and the model CON-FIN that set the coefficient $\mu_i=0.2$ for companies in the banking and financial sector and $\mu_i = 0.05$ for others. The same applies to the mix and conditional models, for which we consider the three variants each: (COND-FLAT), (COND-BNK), (COND-FIN), (MIX-FLAT), (MIX-BNK), and (MIX-FIN). Alternatively, we could have also parametrized $\mu_i$ to some market based measure of systemic relevance, such as the  the connectedness measures computed by network measures computed by \cite{diebold2012better} using volatility spillover networks. Such measures however are typically computed on historical data, while we wanted to use forward looking measures.

In preliminary analyses we also tested richer parametrizations for the contagion models that included one or two additional parameters to control the size of $\mu_i$ across sectors, but we maintained the proposed parametrization due to the satisfactory empirical fit, and model parsimony. In Appendix \ref{sec:mus} we test the pricing of iTraxx tranches using alternative values of $\mu_i$ to assess the effect of different parameter choices on pricing errors.

\section{Simulation study}

We present a simulation study to highlight the characteristics of the proposed models, and to compare the resulting loss distributions. For simplicity we consider a homogeneous setting with uniform marginal default probabilities $p_i = p, \; \forall i$ and parameters across the assets in the portfolio. We also omit the dependency of the parameters to time to simplify the notation.

Figure \ref{fig:sens_p_mu_omega} represents sensitivity of the loss distribution of the CON model to the parameters $\tilde{p}$ (the marginal default probability of each asset), $\omega$ (the proportion of the marginal default probability of each name attributable to contagion), and $\mu$ (the parameter that regulates the infectivity of the defaults). We underline that the interpretation of the parameter $\mu$ is not straightforward. Indeed, not only it regulates the infectivity of the defaults (by changing the parameters $v_i$), but it also indirectly affects the immunization since the parameter $u_i$ in Eq. \ref{eq:calibration}, is set in such a way to preserve the marginal default probability. In contrast, the parameter $\omega$, although it also influences $u_i$ indirectly, retains a straightforward interpretation as the proportion of default probability attributable to contagion. 

\begin{figure}[!h]
	\centering
	\includegraphics[width=1\textwidth]{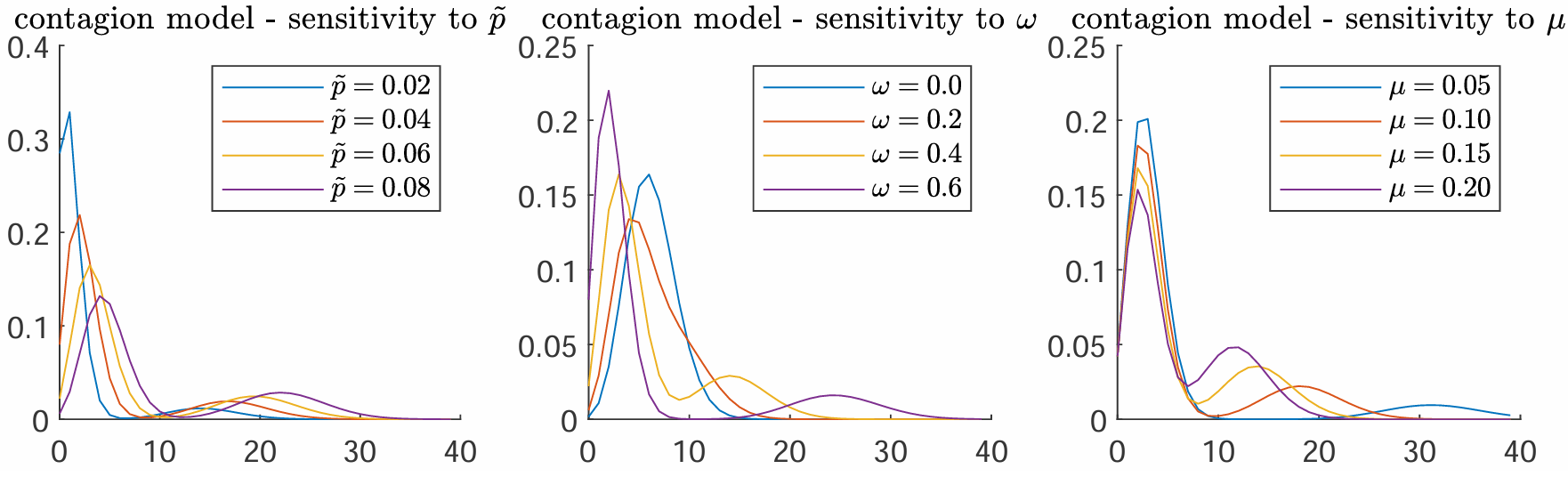}
	\caption{Sensitivity of the loss distribution to parameters $\tilde{p}$, $\omega$, $\mu$. The baseline parameters for the distribution are $\tilde{p} = 0.05$, $\omega=0.5$, $\mu = 0.1$, and there are 125 names in the portfolio. The horizontal axis represents the number of defaults and, although the maximum possible number is 125, only the range up to 40 is displayed for clarity. Parameters are homogeneous across all names. }
	\label{fig:sens_p_mu_omega}
\end{figure} 

We see that overall, the resulting distributions are bimodal, in line with the presence of two distinct possible outcomes depending on the presence or absence of an active infection. Specifically, in the first panel we see the effect of the parameter $\tilde{p}$, finding that the increase of the marginal default probability shifts the distribution to the right. In the middle panel we notice that an increase of the value of $\omega$ amplifies the bimodality of the distribution separating the peaks. Finally, on the right we see that changes in the value  of $\mu$ affect mostly the second peak of the distribution without affecting significantly the main peak. Overall, the contagion model allows to account for flexible shapes of the loss distribution. We point out that there is the potential for identifiability issues due to the similar effects of the parameters $\mu$ and $\omega$ on the loss curve. Thus, in the empirical analysis, $\mu_i$ was held constant to ensure model parsimony.

\begin{table}[!h]
	\centering
	\begin{tabular}{ c | c c c  ccl}
		\hline
		Model      & $\tilde{p}$  &  $\omega$ & $\mu$ & $\rho$& $\pi$& n\\ \hline
		CON  & 0.05 & 0.6 & 0.1 & - & - & 125 \\
		OFG   & 0.05 & - & - & 0.28 & - & 125 \\
		COND & 0.05 & 0.4 & 0.1 & 0.175 & - & 125 \\  
		MIX   & 0.05 & 0.6 & 0.1 & 0.28 & 0.5 & 125 \\ \hline
	\end{tabular}
	\caption{Parameters used for the model comparison: $\tilde{p}$ (marginal default probability), $\omega$ (proportion of marginal default probability attributable to infection in contagion model), $\mu$ (parameter that regulates the infectivity of each name, see Eq. \ref{eq:calibration}), $\rho$ (correlation of the latent variables in the OFG model), $\pi$ (the mixing parameter for the MIX model) and $n$ (the number of assets in the portfolio). Parameters are homogeneous across all names.}
	\label{tab:loss_distr_parameters}
\end{table}
	
Next, we compare the loss distributions under different models: pure contagion (CON), one factor Gaussian copula (OFG), conditional model (COND), and mixture model (MIX). To make the comparison fair, we set the parameters of each distribution in such a way that the expected losses (EL) and the pairwise default correlation between any two assets ($\rho(Z_i,Z_j)$) are the same for each model. LGD is set to 100\%, and the parameters for each distribution are reported in Table \ref{tab:loss_distr_parameters}. In particular, we set a baseline parametrization for the CON model, and we adjusted the correlation parameter $\rho$ in the models. Default correlation is computed for CON, COND, and MIX models as described in Section \ref{sec:def_correl}, while for the OFG model we computed them analytically as described e.g. in \cite{vasicek2002distribution}.

\begin{table}[!h]
	\centering
	\begin{tabular}{ c | c c c  c}
		\hline
		Model      & $\rho(Z_i,Z_j)$ (\%) &  EL (\%) & UL (\%) & Credit VaR$_{0.95}$ (\%)\\ \hline
		CON  &            9.7        &     5.0  &     7.1 &            25.6       \\  
		OFG   &            9.7        &     5.0  &     3.8 &            17.6       \\  
		COND &            9.7        &     5.0  &     7.0 &            32.0       \\
		MIX   &            9.7        &     5.0  &     5.7 &            24.0       \\    \hline
	\end{tabular}
	\caption{Comparison of homogeneous portfolio models: Pairwise default correlation ($\rho(Z_i,Z_j)$), expected loss (EL), unexpected loss (UL), Credit VaR$_{0.95}$.}
\label{tab:loss_distr}
\end{table}

\begin{figure}[!h]
	\centering
	\includegraphics[width=0.8\textwidth]{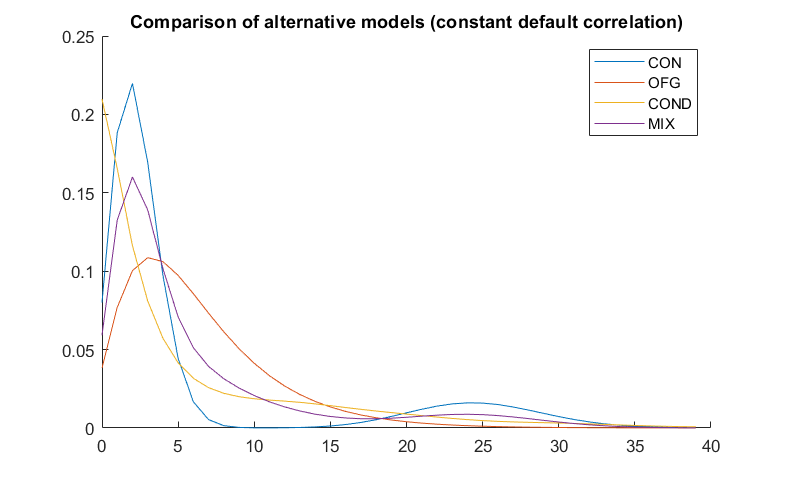}
	\caption{Loss distribution for different models with the parameters in Table \ref{tab:loss_distr_parameters}. The horizontal axis represents the number of defaults and, although the maximum possible number is 125, only the range up to 40 is displayed for clarity. }
	\label{fig:loss_distr}
\end{figure} 

Figure \ref{fig:loss_distr} shows the loss distributions, allowing to compare the models. OFG is characterized by a skewed shape, and a relatively thin tail. In contrast, both CON and MIX exhibit a bimodal shape (in the latter case the bimodality is less pronunciated due to the mixing with the OFG curve). Finally, the conditional model COND has a thicker tail compared to the OFG distribution. Table \ref{tab:loss_distr} reports descriptive statistics of the distribution: Pairwise correlation of the assets ($\rho(Z_i,Z_j)$); EL (Expected Losses, i.e. the proportion of the initial capital lost due to defaults); UL (Unexpected Losses, i.e. the standard deviation of the losses); and Credit VaR (Credit Value at Risk, i.e. the potential maximum losses at a given confidence level). The models, despite having the same pairwise default correlations and expected losses, have different levels of tail risk, with the CON and COND having the highest unexpected losses (UL) and Credit VaR. Showing their capabilities of modelling portfolios with higher tail risk.

\section{Empirical analysis}

\subsection{Dataset}

We consider synthetic CDO tranches of the iTraxx Europe index, that is the main reference for the European market and it is the largest in terms of CDS index tranches\footnote{The notional traded in iTraxx Europe index tranches in 2024 was USD 99.1 billion, ranking above the CDX North America Investment Grade, despite the decrease in volumes compared to the previous year (see \textit{S\&P Dow Jones Indices - 2024 \& H1 2025 Fixed Income Index Products Report}, \url{https://www.spglobal.com/spdji/en/documents/commentary/market-commentary-fixed-income-index-products-annual-report-2024-2025-h1.pdf}).}
Tranches are quoted in terms of upfront paid by the protection buyer to enter in a contract with 100 basis point coupon. The main index is instead listed in terms of par spread (i.e. the level of the premium that balances the expected value of the two legs of the contract). The quotes for contracts with 5 years maturity have been downloaded from LSEG Workspace for the index and for four tranches: Equity (0-3\%), Junior Mezzanine (3-6\%), Senior Mezzanine (6-12\%), Senior (12-100\%). For each day we also download from LSEG Workspace the local 5Y CDS spread for the 125 constituents of the index to compute the marginal default probabilities (used to estimate the marginal default probabilities $\tilde{p}_i$), the industry sector to define $\mu_i$ (for CON-BNK we consider companies whose industry sector is ``Banking'', while for CON-FIN we consider all the companies belonging to the ``Banking'', ``Finance'', or ``Insurance'' industry sectors), and the bootstrapped EURIRS swap rates as the risk free spot rates. We assume a 40\% recovery rate following market convention. Every 6 months a new series is released, so that the index composition reflects the market characteristics. We download the dataset with daily frequency from October 31, 2019 to May 30, 2025, and for each date we consider the most recent iTraxx series available. Overall we considered 1457 trading days, excluding 2.8\% of the dates from the analysis due to missing quotes for the CDO tranches.

\subsection{Pricing with contagion and mix models}

We first report a comparison between the ten proposed models (OFG, CON-FLAT, CON-BNK, CON-FIN, COND-FLAT, COND-BNK, COND-FIN, MIX-FLAT, MIX-BNK, MIX-FIN), for three specific dates. We report the quotes for the CDO tranches and for the entire index, together with the mean error.

For the calibration of the model we use MATLAB 2024a, minimizing an objective function consisting in the sum of root squared percentage error of the quotes (upfronts  for the tranches, par-spread for the whole index). Since some of the  quotes were  very close to zero, we added a translation of 0.1 in the quotes used in the objective function to avoid denominator close or equal to zero to improve the numerical stability of the optimization. The objective function is minimized using \texttt{fmincon} from the Optimization Toolbox, and the initial parameters supplied to the solver are $\omega_0 = \rho_0 = \pi_0 = 0.5$. Each optimization of the COND model (the most computationally intensive, since the contagion loss distribution must be computed for each value of the conditioning variable) takes typically less than two minutes (single core) on a Macbook air with M1 processor and 16 GB of RAM. We bound all the parameters between 0.05 and 0.95 to improve the numerical stability, and, with the exception of the OFG model, in the majority of times we do not obtain corner solutions.\footnote{The code used to estimate the loss curve and to price CDOs under the CON, OFG, COND, and MIX models is available from the corresponding author upon request, together with a toy dataset containing the required inputs in the appropriate format.}

\begin{table}[!h]
	\centering
		\begin{tabular}{ c | c c c c | c | c}
		\hline
		\multicolumn{7}{c}{Calibrated CDO quotes - 30/03/2020} \\ \hline
		Model  & 0-3\% (\%) & 3-6\% (\%) & 6-12\% (\%) & 12-100\% (\%) & index (bps) & error (MAE)\\ \hline
OFG          & 28.38  & 18.46 & 12.85 & -1.46 & 118.23 & 12.63\\ \hline
CON-FLAT     & 37.55  &  6.34 &  6.26 & -1.70 & 103.91 &  6.46\\
CON-BNK      & 43.13  &  5.44 &  5.21 & -1.85 & 102.85 &  5.46\\
CON-FIN      & 39.18  &  6.12 &  6.00 & -1.75 & 103.74 &  6.09\\ \hline
COND-FLAT    & 32.48  & 10.62 &  5.89 & -2.37 &  91.07 &  3.85\\
COND-BNK     & 35.86  & 11.11 &  5.20 & -2.46 &  91.00 &  2.90\\
COND-FIN     & 33.46  & 10.77 &  5.71 & -2.39 &  91.12 &  3.59\\ \hline
MIX-FLAT     & 36.69  & 13.26 &  4.31 & -2.64 &  88.53 &  2.04\\
MIX-BNK      & 39.84  & 12.55 &  4.12 & -2.75 &  87.69 &  \textbf{1.04}\\
MIX-FIN      & 37.58  & 13.05 &  4.25 & -2.66 &  88.49 &  1.80\\ \hline
Market quotes& 42.16  & 12.15 &  4.13 & -2.78 &  85.22 &     -\\ \hline
				\multicolumn{7}{c}{Calibrated CDO quotes - 30/06/2021} \\ \hline
OFG          & 12.43  &  5.77 &  2.55 & -3.82 &  46.86 &  3.66\\ \hline
CON-FLAT     & 43.37  & 10.17 &  3.11 & -5.06 &  46.89 &  6.82\\
CON-BNK      & 18.89  &  0.62 &  0.62 & -3.78 &  46.02 &  1.72\\
CON-FIN      & 43.98  &  9.54 &  3.09 & -5.06 &  46.89 &  6.81\\ \hline
COND-FLAT    & 14.44  &  1.55 &  0.86 & -3.56 &  44.98 &  2.72\\
COND-BNK     & 17.75  &  1.46 &  0.51 & -3.62 &  45.44 &  1.90\\
COND-FIN     & 15.23  &  1.52 &  0.79 & -3.57 &  45.17 &  2.51\\ \hline
MIX-FLAT     & 23.44  &  2.11 & -1.25 & -3.76 &  44.69 &  \textbf{0.55}\\
MIX-BNK      & 23.55  &  2.06 & -1.22 & -3.78 &  44.47 &  0.63\\
MIX-FIN      & 23.45  &  2.11 & -1.24 & -3.76 &  44.65 &  0.56\\ \hline
Market quotes& 23.09  &  2.16 & -1.38 & -3.85 &  46.80 &     -\\ \hline
		
		\multicolumn{7}{c}{Calibrated CDO quotes 30/09/2022} \\ \hline
OFG          & 37.02  & 23.51 & 15.75 & -1.30 & 136.15 &  5.72\\ \hline
CON-FLAT     & 57.95  & 14.03 & 12.82 & -1.56 & 134.84 &  2.83\\
CON-BNK      & 66.04  & 14.34 & 11.20 & -1.74 & 134.73 &  4.08\\
CON-FIN      & 60.18  & 14.00 & 12.40 & -1.61 & 134.83 &  3.21\\ \hline
COND-FLAT    & 51.12  & 18.20 & 12.41 & -1.61 & 131.65 &  2.18\\
COND-BNK     & 54.58  & 19.39 & 11.00 & -1.72 & 130.72 &  1.18\\
COND-FIN     & 52.27  & 18.49 & 12.03 & -1.64 & 131.56 &  1.84\\ \hline
MIX-FLAT     & 55.52  & 20.68 &  9.78 & -1.69 & 131.15 &  \textbf{0.82}\\
MIX-BNK      & 56.40  & 20.21 & 10.08 & -1.94 & 126.86 &  2.00\\
MIX-FIN      & 55.76  & 20.58 &  9.80 & -1.78 & 129.52 &  1.23\\ \hline
Market quotes& 54.61  & 20.93 &  9.93 & -1.56 & 133.81 &     -\\ \hline
		\multicolumn{7}{c}{Calibrated CDO quotes 31/03/2025} \\ \hline
OFG          & 16.36  &  8.88 &  5.08 & -3.00 &  65.91 &  4.74\\ \hline
CON-FLAT     & 21.90  &  3.16 &  3.15 & -2.93 &  64.70 &  2.57\\
CON-BNK      & 25.85  &  2.65 &  2.64 & -3.01 &  64.66 &  1.75\\
CON-FIN      & 23.34  &  3.00 &  2.99 & -2.96 &  64.69 &  2.28\\ \hline
COND-FLAT    & 21.11  &  3.94 &  2.82 & -2.98 &  62.98 &  2.48\\
COND-BNK     & 24.42  &  3.92 &  2.36 & -3.05 &  63.18 &  1.68\\
COND-FIN     & 22.35  &  3.93 &  2.67 & -3.01 &  63.09 &  2.18\\ \hline
MIX-FLAT     & 29.31  &  4.90 &  0.52 & -3.14 &  63.00 &  \textbf{0.21}\\
MIX-BNK      & 29.45  &  4.85 &  0.55 & -3.18 &  62.37 &  0.37\\
MIX-FIN      & 29.34  &  4.89 &  0.53 & -3.15 &  62.83 &  0.25\\ \hline
Market quotes& 29.17  &  4.90 &  0.49 & -3.22 &  63.81 &     -\\ \hline

	\end{tabular}
	\caption{iTraxx tranche calibrated quotes. The quotes for the CDO tranches are upfronts (running spread equal to 100 bps) while the ones for the index are par spreads. The lowest Mean Absolute Error (MAE) is highlighted in bold.}
	\label{tab:Calibration}
\end{table}

Table \ref{tab:Calibration} reports the calibrated quotes, as well as the mean absolute errors in four selected dates for the models. We present the results for March 30, 2020 (at the beginning of the spread of COVID pandemic), June 30, 2021 (after the COVID pandemic, and before the Ukranian war and the rise of interest rates), September, 2022 (in a period of high interest rates, strong geopolitical tension due to the war in Ukraina, and high energy prices), and March 31, 2025 (a period marked by relatively low CDS spreads, yet still affected by several factors undermining global financial stability, including the wars in Ukraine and Gaza, as well as ongoing discussions about tariffs). We see that the MIX models show in all cases the best performance, with  very small errors for the tranches and the index. Concerning the choice of the specific MIX model (FLAT, BNK, or FIN), the results are not clear cut, as the ranking changes across the dates presented. As expected, the OFG show significantly worse performance, but we see that also the models with only the contagion component (CON) do not allow to accurately price the derivatives, meaning that the correlation component is also necessary to properly characterize credit default markets. Finally, the COND models ranks halfway between the CON and MIX models in terms of accuracy. From a practical perspective, the superior performance of the MIX models relative to the conditional ones is advantageous, as the former are less computationally intensive.

\begin{figure}[!h]
	\centering
	\includegraphics[width=1\textwidth]{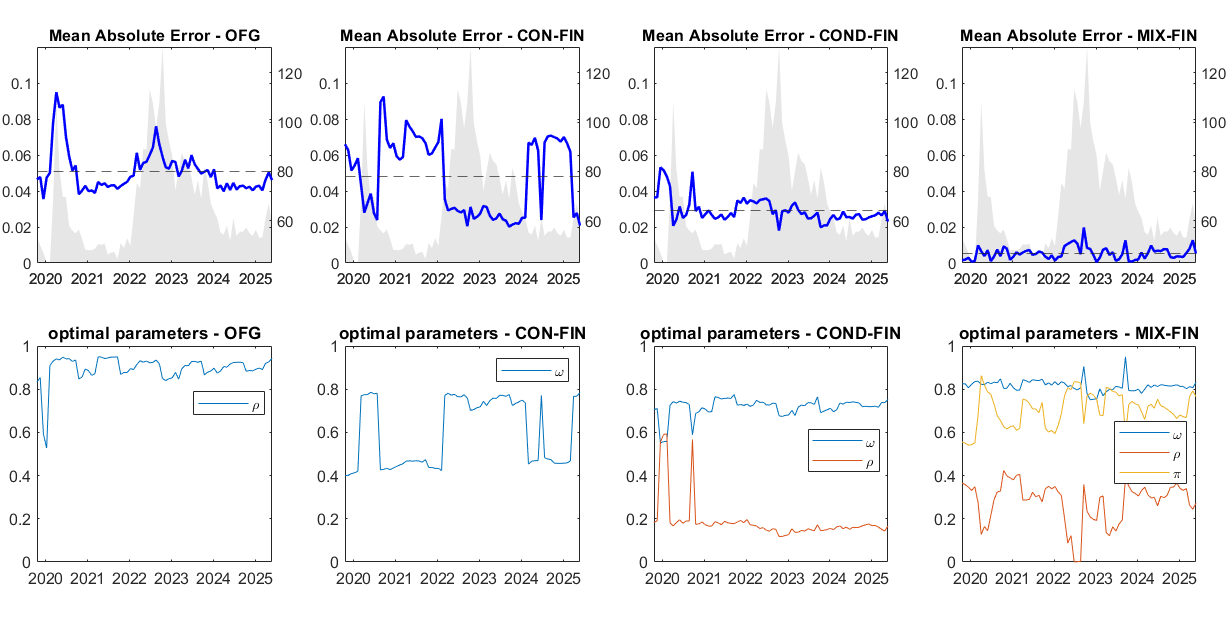}
	\caption{Evolution over time of the average Mean Absolute Error across the priced tranches with spread of the untranched index shaded in the background (upper panels) and optimal parameters (bottom panels) for four selected models (OFG, CON-FIN, COND-FIN, and MIX-FIN models from October 31, 2019 to May 30, 2025.}
	\label{fig:time_series}
\end{figure}

Focusing for brevity on the FIN models, we further study their dynamics by analyzing the evolution over time of the quality of the fit, and the stability of the parameters. Figure \ref{fig:time_series} shows the the Mean Absolute Error (MAE) of the fit for the models from October 31, 2019 to May 31, 2025 (top panels), and  the time series of the optimal parameters (bottom panels). In order to simplify the interpretation of the results, we overlay in gray the spread of the untranched iTraxx index (in basis points), representing the overall dynamics of credit risk in the market.
We see that the MAE for the OFG model follows a dynamic similar to the one of the iTraxx index, meaning that the pricing error of the model is higher when the distress in credit market is high. In contrast, the CON model has a very sharp inverse relation to the iTraxx: in periods of high spreads the pricing error of the pure contagion model is relatively low, but when the credit spreads reduce, the error of the model increases. The performance of the MIX model are consistently good across the period, with a small increase in the error corresponding to the most distressed credit market in 2022, remaining however the best overall model. Finally, the COND model is consistent in terms of MAE across the entire period, but has overall worse performance compared to the FIN model. From an economic perspective, the dynamics of the OFG and CON models in terms of MAE -- together with the superior performance of the MIX models compared to the COND models -- suggests the presence of a regime switch in the market, corroborating the point of using the MIX model.

Concerning the parameters, $\rho$ in the OFG is often close or on the upper boundary (0.95), suggesting that the Gaussian model is not suitable to account for the level of joint default probability implied by the iTraxx tranches market prices. The  parameter of the CON model instead jumps significantly, following the same dynamic of the MAE. Concerning the MIX model, we observe a negative correlation between the dynamics of $\rho$ and $\pi$, while $\omega$ remains high and relatively constant across the entire period. The parameter $\pi$ is always above 0.5, meaning that the ``contagion state'' is predominant over the ``correlated defaults state''. Finally, for the COND model we see that the parameters $\omega$ and $\rho$ tend to be relatively stable and positively correlated, with the exception of two periods of instability at the beginning of the sample.  In Appendix \ref{sec:mus} we test the sensitivity of the models to the parameters $\mu_i$ (that are not calibrated in our specification), finding that the MIX models are the least sensitive to changes, while CON and COND models are more affected by perturbations of such parameters.

To study the models' adaptive boundaries in different market environments, we report the average error in a set of crisis scenarios. In particular, we consider the following 6 scenarios: COVID-19 Sharp Shock (Mar 1 -- Jun 1, 2020), COVID-19 Recovery \& Adaptation	(Jun 1 -- Dec 31, 2020), Post-Pandemic (Jan 1 -- Dec 31, 2021), Ukraine / Energy Crisis (Feb 24 -- Aug 31, 2022), European Interest Rate Tightening (Jul 27 -- Oct 27, 2022), Post-Crisis Normalization (Jan 1, 2023 -- May 31, 2025), and we fit the model each day. Table \ref{tab:periods} reports the average MAE of the days in each period and the corresponding volatility. We see that the MIX models are the best performing in each scenario. As we previously noticed from Figure \ref{fig:time_series}, the OFG tends to perform relatively better in quiet periods, while the contagion models (CON) show lower errors in the harshest periods. We point out that such observation may have implications on the  interpretation of the model: the market perception of the loss curve may be closer to the CON model in distressed period with high perceived probability of default clustering, while it may be closer to the one implied by the Gaussian model in quiter periods. Thus, the proposed parametrization of the models may be used to gauge the market perception of the stability of the market: a higher-than-usual level of $\pi$ in the MIX model, may for instance be studied as a market-based indicator of an increased fear of financial contagion. Still, the focus of this paper is on static loss-distribution comparison. Evaluating model performance in a genuine time-series forecasting setting (e.g., temporal clustering or regime-switching predictability) and formally study the application of the framework to systemic risk analysis is outside the scope of this work and is left for future research.

\begin{table}[!h]
	\centering
	\begin{tabular}{ c | C{1.5cm} C{1.5cm} C{1.5cm} C{1.5cm} C{1.5cm} C{1.5cm}}
			\hline
		Model  & COVID Sharp & COVID Tail & Post Pandemic & Energy Crisis & Int. Rate Rise & Post-Crisis\\ \hline
OFG       & 8.32 & 4.97 & 4.22 & 6.03 & 6.28 & 4.76\\
& (2.24) & (1.03) & (0.45) & (0.79) & (0.69) & (0.76) \\ \hline \hline
CON-FLAT  & 3.51 & 5.68 & 6.93 & 3.14 & 2.89 & 4.23\\
& (1.45) & (2.53) & (0.83) & (1.01) & (0.49) & (2.08) \\ \hline
CON-BNK   & 3.37 & 2.18 & 2.15 & 2.63 & 3.30 & 1.98\\
& (1.05) & (0.39) & (0.43) & (0.68) & (0.82) & (0.39) \\ \hline
CON-FIN   & 4.24 & 7.43 & 6.94 & 3.91 & 2.97 & 5.27\\
& (2.65) & (2.55) & (0.83) & (2.13) & (0.46) & (2.21) \\ \hline \hline
COND-FLAT & 2.96 & 2.96 & 3.00 & 3.37 & 2.77 & 2.66\\
& (0.75) & (0.51) & (0.55) & (0.45) & (0.67) & (0.38) \\ \hline
COND-BNK  & 2.28 & 2.24 & 2.29 & 2.51 & 1.94 & 2.03\\
& (0.77) & (0.35) & (0.40) & (0.44) & (0.62) & (0.36) \\ \hline
COND-FIN  & 3.25 & 3.75 & 3.52 & 3.63 & 3.08 & 2.91\\
& (0.74) & (0.90) & (0.83) & (0.45) & (0.65) & (0.38) \\ \hline \hline
MIX-FLAT  & \textbf{1.02} & 0.54 & 0.51 & \textbf{0.99} & 1.09 & 0.56\\
& (0.65) & (0.40) & (0.24) & (0.54) & (0.57) & (0.39) \\ \hline
MIX-BNK   & 1.10 & 0.66 & 0.58 & 1.52 & 1.67 & 0.66\\
& (0.72) & (0.41) & (0.26) & (0.72) & (0.68) & (0.43) \\ \hline
MIX-FIN   & 1.10 & \textbf{0.50} & \textbf{0.50} & 1.03 & \textbf{1.05} & \textbf{0.54}\\
& (0.64) & (0.39) & (0.26) & (0.49) & (0.55) & (0.38) \\ \hline
		\end{tabular}
\caption{Average mean absolute error (MAE) and its volatility in six periods: COVID-19 Sharp Shock (Mar 1 -- Jun 1, 2020), COVID-19 Recovery \& Adaptation	(Jun 1 -- Dec 31, 2020), Post-Pandemic (Jan 1 -- Dec 31, 2021), Ukraine / Energy Crisis (Feb 24 -- Aug 31, 2022), European Interest Rate Tightening (Jul 27 -- Oct 27, 2022), Post-Crisis Normalization (Jan 1, 2023 -- May 31, 2025). The lowest average MAE for each period is highlighted in bold.}
\label{tab:periods}
\end{table}

Overall, the empirical evidence on the analyzed iTraxx sample suggests that the MIX model, despite not accounting for the possible interaction between factor model and contagion, shows the highest potential for applications, thanks to the good performance and the best computational efficiency compared to the COND model (indeed, despite having 2 rather than 3 parameters, the COND model calibration is more challenging since the contagion loss distribution has to be calculated for each value of the conditioning factor).

\section{Conclusions}
In this paper we presented a new contagion model for the estimation of the loss distribution of a credit portfolio. The model that we introduce is related to the seminal work of \cite{davis2001infectious} and it can achieve reasonable performance with heterogeneous portfolios.

We then introduced two alternative models for the construction of the loss distribution that incorporate the effect of a common factor, either using a mixture distribution or by conditioning on a common Gaussian factor. We apply the model to the problem of pricing synthetic CDO tranches of the iTraxx index, calibrating several specifications of the model on multiple dates with satisfactory performance on the studied dataset. In particular the best performance are obtained using the three parameters mixture models. The models are computationally efficient, and the optimal parameters have a clear economic interpretation.

In addition to fixed-income trading desks, the model has potential applications in the monitoring of systemic risk, as it allows to extract market-based forward looking information on the possible manifestation of systemic events, distinguishing between contagion events and joint defaults due to common factors. We leave the exploration of policy applications and the testing of the empirical properties of the model in other markets (e.g. the \textit{CDX North America Investment Grade} index) to future works.

\newpage

\bibliographystyle{apalike}
\bibliography{Contagion_CDO_googlescholar.bib}

\newpage

\appendix

\section{Theoretical results}\label{A}

Some useful results: at first, we will explore single name default probability under the model assumptions. We will then move to calculate the portfolio loss distribution. A few necessary tools and additional notational shortcuts will be introduced along the way.\newline

\begin{proposition}\label{pr: IA}
	Let $A\subseteq\{1\cdots,n\}$; the probability that at least one name in $A$ spreads an infection is given by the quantity $I_A$ defined as 
	\begin{equation} I_A := \left\{1-\prod_{j\in A} \left[1-p_j(t)\cdot v_j\right]\right\}. \end{equation}
\end{proposition}
\begin{proof}
	[Proof]
	The probability that an infection starts from inside $A$ is given by
	\begin{equation} P\{X_j \cdot V_j=1 \text{ for at least one } j\in A\} = 1- P\{X_j\cdot V_j=0, \forall j \in A\}.\end{equation}
	We can now use the independence assumption on $X$ and $V$ to get
	\begin{equation}
		\begin{array}{lr}
			P\{X_j\cdot V_j=0, \forall j \in A\} & = \\
			\prod_{j \in A} P\{ X_j\cdot V_j =0\} & = \\
			\prod_{j \in A} [1-P\{ X_j\cdot V_j =1\} ] & = \\
			\prod_{j \in A} [1-P\{ X_j=1, V_j =1\} ] & = \\
			\prod_{j \in A} [1-P\{ X_j=1\}\cdot P\{V_j =1\} ] & = \\
			\prod_{j \in A} [1-p_j \cdot v_j ]. & \
		\end{array}
	\end{equation}
	
\end{proof}
We will now focus on single name properties. Let $\overline{A}$ be the complement of set $A$. The following result proves Proposition \ref{pr: I}, giving the formula for the probability of default of single names:

\begin{proposition*}[1]\label{pr: I_A}
	Let $X_i,V_i,U_i$, $i=1,...,n$ be mutually independent Bernoulli variables with probabilities  $p_i= P\{X_i=1\}, \,\, u_i=P\{U_i=1\} , \,\, v_i=P\{V_i=1\}$ on a time horizon $[0,t]$. Let $Z_i$ be defined according to Eq. \ref{eq:MainModDesc} and let $\tilde{p}_i := P\{Z_i=1\}$. We have:
	\begin{equation}\label{eq: PD_A} \tilde{p}_i = p_i+[1-p_i]\cdot [1-u_i]\cdot I_{\overline{\{i\}}},\end{equation}
	where 
	\begin{equation} I_{\overline{\{i\}}} := \left\{1-\prod_{j\neq i } \left[1-p_j\cdot v_j\right]\right\}. \end{equation}
\end{proposition*}

\begin{proof}
	We have that
	\begin{equation}
		\begin{array}{cl} 
			P\{Z_i=1\} = & P\{Z_i=1 | X_i=1\} \cdot P\{X_i=1\} + \\
			\quad & P\{Z_i=1| X_i=0\} \cdot P\{X_i=0\}. \\
		\end{array}
	\end{equation}
	It is easy to see that
	\begin{align*}
		\begin{array}{c} 
			P\{Z_i=1| X_i=1\}=1, \\
			P\{X_i=1\}=p_i, \\
			P\{X_i=0\}=1-p_i,
		\end{array} 
	\end{align*}
	so that the only part left to calculate is $P\{Z_i=1| X_i=0\}$; using the expression for $Z_i$ in \eqref{eq:MainModDesc} we have
	\begin{equation}P\{Z_i=1| X_i=0\} = P\left\{(1-U_i) \cdot \left[ 1 - \prod_{i\neq j} (1-X_j \cdot V_j ) \right] = 1\right\}. \end{equation}
	Both variables in the last expression can only take binary values $(0,1)$ so their product can only be $1$ if they both take value $1$. This implies that 
	\begin{equation}\begin{array}{l} 
			P\left\{(1-U_i) \cdot \left[ 1 - \prod_{i\neq j} (1-X_j \cdot V_j) \right] = 1\right\} \\
			= P\left\{(1-U_i)=1, \left[ 1 - \prod_{i\neq j} (1-X_j \cdot V_j) \right] = 1\right\}.
	\end{array} \end{equation}
	We can use the independence assumption between the various building blocks to split the right hand side as 
	\begin{equation}P\left\{(1-U_i)=1\right\} \cdot P\left\{\left[ 1 - \prod_{i\neq j} (1-X_j \cdot V_j) \right] = 1\right\}, \end{equation}
	and hence, thanks to proposition \ref{pr: IA} we have:
	\begin{equation}P\{Z_i=1| X_i=0\} = (1-u_i) \cdot I_{\overline{\{i\}}},\end{equation}
	that concludes the proof.\end{proof}

Intuitively, name $i$ can default in two ways: idiosyncratically (with probability $p_i$) or by contagion if it survives ($1-p_i$), fails to defend itself ($1-u_i$) and an external infection is active ($I_{\overline{i}}$). Let us move now to the proof of proposition \ref{portfolio}:

\begin{proposition*}[2]
	Consider a portfolio comprising $n$ assets, with each asset denoted by the index $j = 1,\dots,n$. The following recursive relationship links $[\alpha_j(\cdot,\cdot), \beta_j(\cdot,\cdot)]$ to $[\alpha_{j-1}(\cdot,\cdot),\beta_{j-1}(\cdot,\cdot)]$:
	\begin{equation}\label{alfabetaA}
		\begin{array}{cclr} 
			\ 								&\  & (1-p_j(t)) \cdot u_j \cdot \alpha_{j-1}(h,k,t) 			& + \\
			\alpha_j(h,k,t) & = & (1-p_j) \cdot (1-u_j(t)) \cdot \alpha_{j-1}(h,k-d_j) & + \\ 
			\ 								&\  & p_j(t) \cdot (1-v_j) \cdot \alpha_{j-1}(h-d_j,k),			& \  \\
			\ 								&\  & \ 											& \  \\
			\ 								&\  & (1-p_j(t)) \cdot u_j(t) \cdot \beta_{j-1}(h,k,t) + p_j(t) \cdot \beta_{j-1}(h-d_j,k) & + \\
			\beta_j(h,k,t) 		& = & (1-p_j(t)) \cdot (1-u_j) \cdot \beta_{j-1}(h,k-d_j) & + \\ 
			\ 								&\  & p_j(t) \cdot v_j \cdot \alpha_{j-1}(h-d_j, k), &\ \\
		\end{array}
	\end{equation}
	with the following boundary conditions:
	\begin{equation}\label{boundarycondA}
		\begin{array}{l}
			\alpha_{0,t}(0,0) = 1, \\ \alpha_{0,t}(i,j) = 0 \quad \quad \forall (i,j) \neq (0,0), \\ \beta_{0,t}(i,j) =0 \quad \quad \forall i,j.
		\end{array}
	\end{equation}

	Moreover, the order chosen to include the terms does not affect the final result.
\end{proposition*}
\begin{proof}
	In order to obtain a set of equations for $\alpha_j(\cdot,\cdot)$, consider that there are 3 ways of reaching $\alpha_j(h,k)$ starting from $\alpha_{j-1}(h,k)$ and adding a new name:
	\begin{enumerate}
		\item \textbf{Full survival}
		\begin{equation}(1-p_j) \cdot u_j \cdot \alpha_{j-1}(h,k). \end{equation}
		The name survives with probability $(1-p_j)$ and protects itself from future aggressions $(u_j)$. No losses are realized neither potential ones added.
		\item \textbf{Partial survival}
		\begin{equation}(1-p_j) \cdot (1-u_j) \cdot \alpha_{j-1}(h,k-d_j). \end{equation}
		The name survives with probability $(1-p_j)$ but fails to protect itself against future aggressions $(1-u_j)$. Its $d_j$ units of losses are at risk should an infection spread.
		\item \textbf{Non-infectious default}
		\begin{equation}p_j \cdot (1-v_j) \cdot \alpha_{j-1}(h-d_j,k). \end{equation}
		The name defaults directly $(p_j)$ but it is not trying to start an infection $(1-v_j)$.
	\end{enumerate}
	$$\quad$$
	Similarly, there are 4 ways of reaching $\beta_j(h,k)$:
	\begin{enumerate}
		\item \textbf{Full survival}
		\begin{equation} (1-p_j) \cdot u_j \cdot \beta_{j-1}(h,k).\end{equation}
		The name survives $(1-p_j)$ and protects itself against the current and future infections $(u_j)$.
		\item \textbf{Default by contagion}
		\begin{equation} (1-p_j) \cdot (1-u_j) \cdot \beta_{j-1}(h,k-d_j).\end{equation}
		The name survives $(1-p_j)$ but fails to protect itself against the existing infection $(1-u_j)$.
		\item \textbf{Direct default}
		\begin{equation}p_j \cdot \beta_{j-1}(h-d_j,k).\end{equation}
		The name defaults $(p_j)$ and in this case we don't need to consider separately the cases in which it spreads or not the infection as we are already in an infected world.
		\item \textbf{First infection}
		\begin{equation}p_j(t) \cdot v_j \cdot \alpha_{j-1}(h-d_j, k).\end{equation}
		The name defaults $(p_j)$ and spreads the contagion $(v_j)$ in a previously uncontaminated world causing the $k$ units of potential losses to become real ones.
	\end{enumerate}
	Putting together the previous equations, we get system \eqref{alfabetaA}.
	
	Let's now prove that the order with which we add names is not important for the final result. We report the proof only for $\alpha$ as the case for $\beta$ is similar. Suppose that we want to add two names, $i$ first and then $j$, to a set of $m$ names. In order to shorten the notation, let us indicate with $\bar{p}=1-p$, $\bar{v}=1-v$ and $\bar{u}=1-u$. When we add name $j$, we would apply \eqref{alfabetaA} to a portfolio of $m+1$ names obtaining
	\begin{equation}\label{indproof1}
		\begin{array}{cclr} 
			\ 								&\  & \bar{p}_j \cdot u_j \cdot \alpha_{m+1}(h,k) 			& + \\
			\alpha_{m+2,t}(h,k,t) & = & \bar{p}_j \cdot \bar{u}_j \cdot \alpha_{m+1}(h,k-d_j) & + \\ 
			\ 								&\  & p_j \cdot \bar{v}_j \cdot \alpha_{m+1}(h-d_j,k).			& \  \\
		\end{array}
	\end{equation}
	Each of the three terms $\alpha_{m+1}$ on the right hand side can be explicitly written by applying \eqref{alfabetaA} again:
	\begin{equation}\label{indproof2A}
		\begin{array}{cclr} 
			\ 								&\  & \bar{p}_i \cdot u_i \cdot \alpha_{m}(h,k) 			& + \\
			\alpha_{m+1}(h,k) & = & \bar{p}_i \cdot \bar{u}_i \cdot \alpha_{m}(h,k-d_i) & + \\ 
			\ 								&\  & p_i \cdot \bar{v}_i \cdot \alpha_{m}(h-d_i,k),			& \  \\
		\end{array}
	\end{equation}
	\begin{equation}\label{indproof2B}
		\begin{array}{cclr} 
			\ 								    &\  & \bar{p}_i \cdot u_i \cdot \alpha_{m}(h,k-d_j) 			& + \\
			\alpha_{m+1}(h,k-d_j) & = & \bar{p}_i \cdot \bar{u}_i \cdot \alpha_{m}(h,k-d_j-d_i) & + \\ 
			\ 								    &\  & p_i \cdot \bar{v}_i \cdot \alpha_{m}(h-d_i,k-d_j),			& \  \\
		\end{array}
	\end{equation}
	\begin{equation}\label{indproof2C}
		\begin{array}{cclr} 
			\ 								&\  & \bar{p}_i \cdot u_i \cdot \alpha_{m}(h-d_j,k) 			& + \\
			\alpha_{m+1}(h-d_j,k) & = & \bar{p}_i \cdot \bar{u}_i \cdot \alpha_{m}(h-d_j,k-d_i) & + \\ 
			\ 								&\  & p_i \cdot \bar{v}_i \cdot \alpha_{m}(h-d_j-d_i,k)	.		& \  \\
		\end{array}
	\end{equation}
	Substituting \eqref{indproof2A}, \eqref{indproof2B} and \eqref{indproof2C} into \eqref{indproof1} and rearranging terms, we can write
	\begin{align*}
		\begin{array}{clr} 
			\alpha_{m+2}(h,k)= & \bar{p}_i \cdot u_i \cdot \bar{p_j} \cdot u_j \cdot \alpha_{m}(h,k) & + \\
			\ 								& \bar{p}_i \cdot \bar{u}_i \cdot \bar{p}_j \cdot \bar{u}_j \cdot \alpha_{m}(h,k-d_i-d_j) & + \\ 
			\ 								& p_i \cdot \bar{v}_i \cdot p_j \cdot \bar{v}_j \cdot \alpha_{m}(h-d_i-d_j,k)			& +\\
			\ 								& \bar{p}_i \cdot u_i \cdot \bar{p}_j \cdot \bar{u}_j \cdot \alpha_{m}(h,k-d_j)& + 	\\
			\
			&\bar{p_j} \cdot u_j \cdot \bar{p}_i \cdot \bar{u}_i \cdot \alpha_{m}(h,k-d_i)& + \\
			\ 								& \bar{p}_i \cdot u_i \cdot p_j \cdot \bar{v_j} \cdot \alpha_{m}(h-d_j,k)& + \\
			\ 
			&\bar{p}_j \cdot u_j \cdot p_i \cdot \bar{v_i} \cdot \alpha_{m}(h-d_i,k)& + \\
			\ 								& \bar{p}_i \cdot \bar{u}_i \cdot p_j \cdot \bar{v_j} \cdot \alpha_{m}(h-d_j,k-d_i)& + \\
			\ 
			&\bar{p}_j \cdot \bar{u}_j \cdot p_i \cdot \bar{v}_i \cdot \alpha_{m}(h-d_i,k-d_j).& \ \\
		\end{array}
	\end{align*}
	Every line in the last equation is symmetric with respect to $i$ and $j$ and then we can invert the order between $i$ and $j$ without changing the final result.\footnote{Technically, this only proves that the order of the last 2 names added is not important. It is quite easy to extend the reasoning to the entire sequence using induction on $m$.}

\end{proof}


\begin{proposition*}[3]\label{pr:NoLosses_A}
	The probability of observing no losses is given by the following result:
	\begin{equation} \label{eq:NoLosses_A}
		P\{ L_n (t)= 0\} = \prod_{j=1}^{n} (1-p_j(t)).
	\end{equation}

\end{proposition*}
\begin{proof}
	The assert can be easily proved by recursion on $n$. We start with the case $n=2$. A direct application of equation \eqref{eq:LossFromAlfaBeta} leads to 
	$$ P\{ L_2(t) = 0\} = \alpha_{2,t}(0,0)+\alpha_{2,t}(0,d_1)+\alpha_{2,t}(0,d_2)+\alpha_{2,t}(0,d_1+d_2).$$
	For each of the four terms on the right hand side we can apply recursively equation \ref{alfabetaA} (keeping in mind \ref{boundarycondA}) and obtain
	$$ \alpha_{2,t}(0,0) = (1-p_1(t))\cdot(1-u_1(t)) \cdot (1-p_2(t))\cdot(1-u_2(t)), $$
	$$ \alpha_{2,t}(0,d_1) = (1-p_1(t))\cdot(u_1(t)) \cdot (1-p_2(t))\cdot(1-u_2(t)), $$
	$$ \alpha_{2,t}(0,d_2) = (1-p_1(t))\cdot(1-u_1(t)) \cdot (1-p_2(t))\cdot(u_2(t)), $$
	$$ \alpha_{2,t}(0,d_1+d_2(t)) = (1-p_1(t))\cdot(u_1(t)) \cdot (1-p_2(t))\cdot(u_2(t)). $$
	Summing the four equations above we get
	\begin{align*}
		\begin{array}{c} 
			P\{ L_2(t) = 0\}=(1-p_1(t))\cdot(1-p_2(t)) \times\\
			\quad \quad \times \left[ (1-u_1(t))(1-u_2(t))+u_1(t)(1-u_2(t))+(1-u_1(t))u_2+u_1(t) u_2(t)\right],
		\end{array}
	\end{align*}
	from which the assert as $$(1-u_1(t))(1-u_2(t))+u_1(t)(1-u_2(t))+(1-u_1(t))u_2(t)+u_1(t) u_2(t)=1.$$
	Assume now that \eqref{eq:NoLosses} holds for $n-1$ and then prove it for $n$. We have
	$$P\{ L_n(t) = 0\} = P\{ L_n(t) = 0 | L_{n-1}(t) = 0\}\cdot P\{ L_{n-1}(t) = 0\}.$$
	The above is true since $P\{ L_n(t) = 0 | L_{n-1}(t) > 0\}=0$. Thanks to the inductive hypothesis we have that $P\{ L_{n-1}(t) = 0\}=\prod_{j=1}^{n-1} (1-p_j(t))$; coming to $P\{ L_n(t) = 0 | L_{n-1}(t) = 0\}$, this represents the probability that a portfolio of $n$ names suffers no losses given the already $n-1$ names experience no defaults. It is then immediate to see that the only way this can be possible is that also the $n$th name does not default idiosyncratically. Exploiting the independence among the various components, we get
	$$P\{ L_n(t) = 0 | L_{n-1}(t) = 0\} = (1-p_n(t)).$$
\end{proof}

	The following result (Proposition \ref{prop:correl}) provides a formula for the pairwise default correlation of two names in an homogeneous setting. We underline that, following an analogous approach, the result can be extended to the heterogeneous case.
	
	\begin{proposition*}[4] \label{prop:correlA}
		Consider a portfolio with $n$ assets in which defaults are regulated by the infection model in Section \ref{sec:contagion} and such that $p_i = p,\; v_i = v,\; u_i = u, \tilde{p}_i = \tilde{p} \; \forall i$. The pairwise default correlation of $Z_i$ and $Z_j$ is:
		
		\begin{equation}
			\rho_{CON}(Z_1,Z_2) = \frac{P(Z_i = 1,Z_j=1) -\tilde{p}^2}{\tilde{p}(1-\tilde{p})},\notag
		\end{equation}
		where $\tilde{p}$ is the marginal default probability of each name (see Eq. \ref{eq: PD_A}), and the probability of joint defaults of asset with indices $i$ and $j$ is:
		\begin{equation}\label{eq:joint_def}
			P(Z_i = 1,Z_j=1) = p^2 + 2p(1-p)(1-u)\left[1-(1-v)(1-pv)^{n-2}\right] + (1-p)^2 (1-u)^2\left[1-(1-pv)^{n-2}\right].\notag
		\end{equation}	
	\end{proposition*}
	
	\begin{proof}
		For two assets with indices $i,j$, the default probabilities are 
		\begin{equation} 
			Z_i = X_i + \left[1-X_i\right] \cdot \left[1-U_i\right] \cdot \left\{ 1 - \prod_{k\neq i} \left[1-X_k \cdot V_k\right] \right\},
		\end{equation}
		\begin{equation} 
			Z_j = X_j + \left[1-X_j\right] \cdot \left[1-U_j\right] \cdot \left\{ 1 - \prod_{k\neq j} \left[1-X_k \cdot V_k\right] \right\}.
		\end{equation}
		
		Then, we consider the following four disjoint cases. For each of the cases, we compute the probability of the case and the conditional joint default probability.

		\begin{enumerate}
			\item 	Both assets default idiosyncratically ($[X_i = 1, X_j = 1]$). We have $$P(X_i = 1, X_j = 1) = p^2.$$ The conditional probability of default is trivially equal to one since both names defaulted idiosyncratically:  $$P(Z_i = 1,Z_j=1 | X_i = 1, X_j = 1) = 1.$$
			\item Only the first asset defaults idiosyncratically ($[X_i = 1, X_j = 0]$). The probability of the case is: $$P(X_i = 1, X_j = 0) = p(1-p).$$ The conditional default probability is: $$P(Z_i = 1,Z_j=1 | X_i = 1, X_j = 0) = (1-u)\left[1-(1-v)(1-pv)^{n-2}\right].$$ Indeed, asset $j$ defaults only if it fails to obtain immunization ($P(U_j = 0) = (1-u)$), and at least one asset transmits contagion (note that in the conditional scenario we have $P(X_i = 1 | X_i = 1, X_j = 0) = 1$, hence the probability that asset $i$ does not generate contagion is $1-v$, and not $1-pv$).
			\item Only the second asset defaults idiosyncratically ($[X_i = 0, X_j = 1]$). The argument is symmetrical to the previous point due to homogeneity, thus $P(X_i = 0, X_j = 1) = P(X_i = 1, X_j = 0) $ and $ P(Z_i = 1,Z_j=1 | X_i = 0, X_j = 1) = P(Z_i = 1,Z_j=1 | X_i = 1, X_j = 0)$.
			\item None of the two assets default idiosyncratically ($[X_i = 0, X_j = 0]$). The probability of the case is $$P(X_i = 0, X_j = 0) = (1-p)^2,$$ and the conditional default probability is $$P(Z_i = 1,Z_j=1 | X_i = 0, X_j = 0) = (1-u)^2 \left[1-(1-pv)^{n-2}\right]$$ since the two assets default contagion only if they both fail to immunize, and if any of the other assets in the system defaults and transmits contagion.
		\end{enumerate}
		
		We then compute the unconditional joint probability as:
		
		\begin{align}
			P(Z_i = 1,Z_j=1) &= P(X_i = 1, X_j = 1) \cdot P(Z_i = 1,Z_j=1 | X_i = 1, X_j = 1) + \notag\\
			& + 2P(X_i = 1, X_j = 0) \cdot P(Z_i = 1,Z_j=1 | X_i = 1, X_j = 0) +\notag \\
			& + P(X_i = 1, X_j = 1) \cdot P(Z_i = 0,Z_j=0 | X_i = 0, X_j = 0)\notag \\
			& =  p^2 + 2p(1-p)(1-u)\left[1-(1-v)(1-pv)^{n-2}\right] + (1-p)^2(1-u)^2\left[1-(1-pv)^{n-2}\right]. \notag
		\end{align}
		
		To complete the proof, we remind that for the homogeneity of the assets the covariance and correlation of $Z_i, Z_j$ is
		
		\begin{equation}
			\mathbb{C}_{CON}(Z_i,Z_j) = P(Z_i=1,Z_j=1) - \tilde{p}^2, \notag
		\end{equation}
		\begin{equation}
			\rho_{CON}(Z_i,Z_j) = \frac{\mathbb{C}_{CON}(Z_i,Z_j)}{\mathbb{V}_{CON}(Z_i,Z_j)}  = \frac{\mathbb{C}_{CON}(Z_i,Z_j)}{\tilde{p}(1-\tilde{p})},\notag
		\end{equation}
		where $\tilde{p}$ is the marginal default probability of each asset.
		
	\end{proof}

\section{Monte Carlo simulation of the loss distribution}\label{sec:MC}

In case of very large portfolio, the construction of the loss distribution using Equation \ref{eq:LossFromAlfaBeta} and Algorithms \ref{algo:alpha} and \ref{algo:beta} can become computational unfeasible. In such cases the loss distribution can be approximated using Monte Carlo simulations. Indeed, the losses can be computed by simulating three independent Bernoulli variables $X_i$, $V_i$, $U_i$ for each of the assets with index $i = 1, \dots, n$ and computing $X_i$ using Equation \ref{eq:MainModDesc}.
Here we compare empirically the loss distribution obtained using Equation \ref{eq:LossFromAlfaBeta} and Monte Carlo simulations. We consider the following constant parameters: $\tilde{p}_i=0.05$, $\omega_i= 0.5$ for $i = 1, \dots, n$. Figure \ref{tab:MC_simulations} reports the run time, showing that the computational time of our algorithm grows substantially with the number of assets. Still, it remains competitive with Monte Carlo, even for $n=750$.

\begin{table}[!h]
	\centering
	\begin{tabular}{ c |c c }
		\hline
		& \multicolumn{2}{c}{Run time (seconds)} \\
		$n$ & Proposed algorithm & Monte Carlo \\ \hline
		50 & 0.004& 0.217\\
		100 & 0.02& 0.540\\
		125 & 0.047& 0.793\\
		150 & 0.071& 1.204\\
		200 & 0.168& 1.879\\
		500 & 2.89& 11.246\\
		750 & 8.54& 26.847\\ \hline
	\end{tabular}
	\caption{Computational time required for the estimation of the loss distribution with $n$ assets. For Monte Carlo we run 5000 scenarios. Results are the averages across 20 runs.}
	\label{tab:MC_simulations}
\end{table}

\begin{table}[!h]
	\centering
	\begin{tabular}{ c |c}
		\hline
		\multicolumn{2}{c}{Monte Carlo approximation} \\ \hline
		$n$ & KL divergence \\ \hline
		1000 & 0.0735\\
		2500 & 0.0165\\
		5000 & 0.0068\\
		10000 & 0.0041\\
		20000 & 0.0019\\
		50000 & 0.0007\\  \hline
		
	\end{tabular}
	\caption{KL divergence between the Monte Carlo estimates of the loss distribution computed with different numbers of scenarios and the true distribution obtained using Equation \ref{eq:LossFromAlfaBeta}. For Monte Carlo we run 5000 scenarios. Results are the averages across 20 runs.}
	\label{tab:MC_simulations2}
\end{table}

\begin{figure}[!h]
	\centering
	\includegraphics[width=0.49\textwidth]{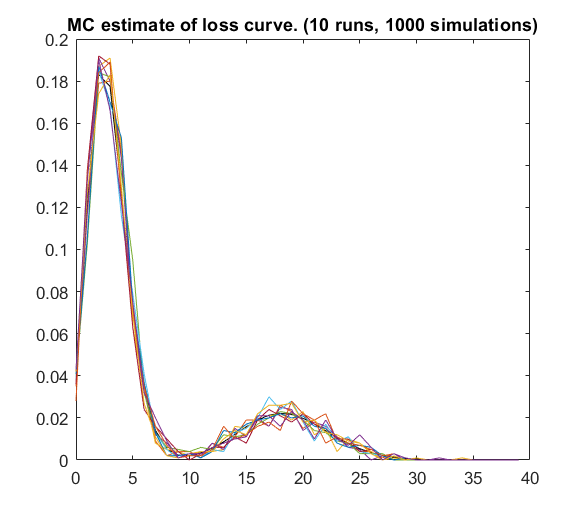}
	\includegraphics[width=0.49\textwidth]{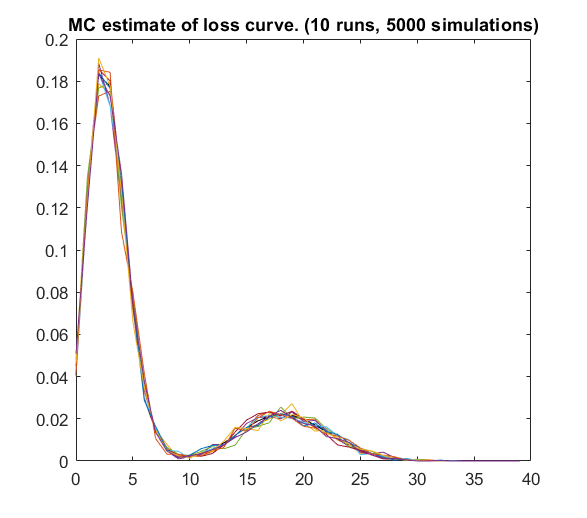}
	\includegraphics[width=0.49\textwidth]{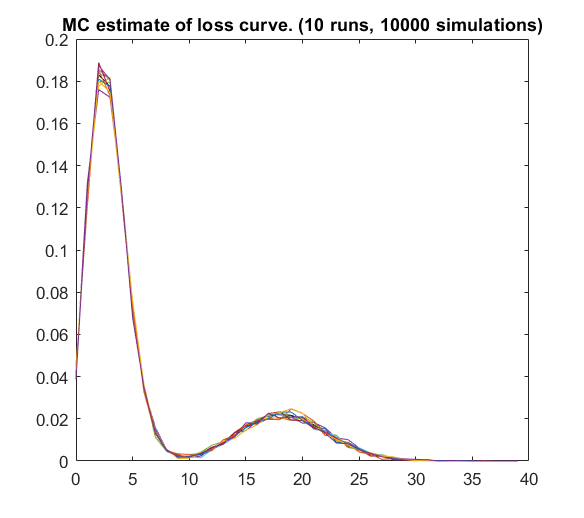}
	\caption{10 Monte Carlo estimates of the loss distribution based on 1000, 5000, and 10000 simulations.}
	\label{fig:MC_loss}
\end{figure}

We then assess the quality of the Monte Carlo approximation of the loss distribution for different number of simulations, fixing the size of the portfolio $n$. In particular, assuming constant LGD equal to 1, we measure the Kullback–Leibler (KL) divergence between the loss distribution estimated with Monte Carlo and the true one obtained with our approach. The KL divergence is defined as:
\begin{equation}
	D_{\mathrm{KL}}(P \,\|\, \hat{P}) = \sum_{h} P(h) \, \log \frac{P(h)}{\hat{P}(h)},
\end{equation}
where $P(h)$ and $\hat{P}(h)$ are the probability of having $h$ losses according to the true distribution, and to the Monte Carlo approximation. Table \ref{tab:MC_simulations2} reports the divergence for different Monte Carlo simulation size. We see that the Monte Carlo estimation converges to the true distribution as $n$ grows. Figure \ref{fig:MC_loss} shows visually 10 Monte Carlo estimates of the loss distribution computed with 1000, 5000, and 10000 simulations, respectively. We see that the estimates with 1000 and 5000 simulations are noisy, while the ones with 10000 simulations are closer to the true one.

\section{Calibration under alternative values of $\mu_i$}\label{sec:mus}

In the empirical analysis, as discussed in Section \ref{sec:choice_mu}, we considered three specifications for the parameters $\mu_i$: in one case they have been kept uniform across companies (FLAT), then we considered higher values for companies in the financial sector and in the banking sector (FIN and BNK, respectively). In this Appendix we test more specifications, repeating the calibration of the models using alternative values of $\mu^*_i$ to assess the sensitivity of the calibration to such parameters. In particular, we set $\mu^*_i =\eta \cdot \mu_i$, where $\mu_i$ is the value of the coefficients for either the FLAT, FIN, or BNK cases (see Section \ref{sec:choice_mu}), with $\eta \in \{ 0.5,0.75,1.25,1.5\}$, and recalibrate the model on the four dates presented in Table \ref{tab:Calibration} (30/03/2020, 20/06/2021, 30/09/2022, 31/03/2025).

Table \ref{tab:mus} shows the average MAE across four dates for different values of $\eta$. We see that the performance of the CON and COND models tend to be slightly better for lower values of $\eta$ compared to the baseline, showing lower estimation error. In contrast, the fit of the MIX models ---the ones with the best fit among the framework tested--- is rather stable across the studied values of $\eta$.

The results suggest that the introduction of $\eta$ as a free parameter in the calibration may be justifiable for the CON and COND model ---although at the cost of increased computational burden and identifiability issues. In contrast, the MIX model seems to be less sensitive to the value of $\eta$.

\begin{table}[!h]
	\centering
	\begin{tabular}{ c | c c c c c}
		\hline
		Model      &$\eta=0.5$&$\eta=0.75$&$\eta=1$&$\eta=1.25$&$\eta=1.5$\\ \hline
CON-FLAT   &  \textbf{0.178} & 0.181 & 0.324 & 0.326 & 0.296 \\
CON-BNK    &  0.207 & \textbf{0.169} & 0.182 & 0.323 & 0.393 \\
CON-FIN    &  0.183 & \textbf{0.175} & 0.274 & 0.394 & 0.282 \\ \hline
COND-FLAT  &  \textbf{0.079} &  0.13 & 0.197 & 0.233 & 0.255 \\
COND-BNK   &  \textbf{0.067} & 0.092 &  0.13 & 0.186 &  0.22 \\
COND-FIN   &  \textbf{0.073} & 0.118 & 0.166 & 0.222 & 0.245 \\ \hline
MIX-FLAT   &  0.065 & 0.054 & \textbf{0.049} & 0.056 & 0.066 \\
MIX-BNK    &  0.062 & 0.057 & 0.055 & \textbf{0.052} & 0.049 \\
MIX-FIN    &  0.068 & 0.055 & \textbf{0.051} & 0.051 & 0.062 \\ \hline
	\end{tabular}
	\caption{Average MAE across four dates. Sensitivity of the calibration to the values of $\mu^*_i$. The calibration is the same as in Table \ref{tab:Calibration}, with the exception of the parameters $\mu_i^* =  \eta \cdot \mu_i $.  The lowest average MAE for each row is highlighted in bold.}
	\label{tab:mus}
\end{table}

\end{document}